%% file: main.tex
\documentclass[
aps, pra,
11pt,
showkeys,
superscriptaddress,
nofootinbib,
floatfix,
]{revtex4-2}
\usepackage{pretex}

\begin{document}
\title{Trainability Beyond Linearity in Variational Quantum Objectives}

\author{Gordon Ma}
\thanks{Corresponding author}
\email{gordonma7@gmail.com}
\email{gordonma@nus.edu.sg}
\affiliation{Centre for Quantum Technologies, National University of Singapore, Singapore}

\author{Xiufan Li}
\email{lixiufan@u.nus.edu}
\affiliation{Centre for Quantum Technologies, National University of Singapore, Singapore}

\date{\today}
\begin{abstract}
\input{abstract}
\end{abstract}


\maketitle

\input{intro_rev_v3}

\input{Theorems/theorem_A_final_rev}

\input{Theorems/theorem_B_rev_final}

\input{Theorems/corollary_C}

\input{04-qml-case-study-final-v2}

\input{05-compressed-interface-v2} 

\input{06-hydro-case-study-rev2}

\input{07-discussion-rev}

\medskip

\begin{acknowledgments}
The authors thank Daniel Leykam for questions and discussions that motivated the numerical demonstration, and Supanut Thanasilp for discussions that sharpened the structural framing and helped motivate Section~V. G.M. is grateful to Dimitris Angelakis for supervision and guidance. This project is supported by the National Research Foundation, Singapore through the National Quantum Office, hosted in A*STAR, under its Centre for Quantum Technologies Funding Initiative (S24Q2d0009).
\end{acknowledgments}

\medskip

\paragraph*{Authors contributions.}
G.M. developed the framework, performed the numerical experiments, and wrote the manuscript. Both authors contributed to the intellectual development of the project.

\newpage

\bibliographystyle{unsrt}
\bibliography{references}

\newpage
\appendix
\input{Appendix/app-iff-rev}
\input{Appendix/app-elliptical}

\input{Appendix/app-04-experiment-details}

\end{document}

%% file: abstract.tex
Barren-plateau results have established exponential gradient suppression as a widely cited obstacle to the scalability of variational quantum algorithms. When and whether these results extend to a given objective has been addressed through loss-specific arguments, but a general structural characterization has remained open. We show that the objective itself admits a fixed-observable representation if and only if the loss is affine in the measured statistics, thereby identifying the exact boundary of the standard concentration-based proof template. Existing transfer results for non-affine losses achieve this reduction under additional assumptions; our characterization implies that such a reduction is not structurally available for a class of non-affine objectives, placing them outside the automatic reach of the existing proof template. Beyond the affine regime, a chain-rule decomposition reveals three governing factors---model responsivity, loss-side signal, and transmittance---and induces a loss-class dichotomy: bounded-gradient losses inherit suppression, while amplification-capable losses can in principle counteract it. In the exponentially wide setting, both classes fail, but for different structural reasons. When the interface is instead designed at polynomial width---exposing coarse-grained statistics rather than individual bitstring probabilities---the exponential-dimensional obstruction is relaxed and the dichotomy plays a genuine role. In a numerical demonstration on a charge-conserving quantum system, the amplification-capable objective produces resolved gradients several orders of magnitude larger than affine and Lipschitz baselines at comparable shot budgets. Over the tested interval, its scaling trend is statistically distinguished from the exponential trend of both alternatives. The boundary is affine; what lies beyond it is a representation-design problem.

%% file: intro_rev_v3.tex
\setlength\epigraphwidth{.45\textwidth}
\setlength\epigraphrule{0pt}
\epigraph{\itshape A cage went in search of a bird.}{---Franz Kafka, \textit{The Zürau Aphorisms}}

\section{Introduction}

Barren-plateau theorems concern gradients of expectation-value objectives \cite{mcclean_barren_2018, cerezo_cost_2021, larocca_barren_2025}. Usually the loss has the form
\begin{equation*}
L(\theta) = \tr\{H\,\rho(\theta)\},
\end{equation*}
for a fixed observable $H$ and each $\theta_k \in[0, 2\pi)$. In this setting each gradient component $\partial_{\theta_k} L$ can be written as the expectation of a fixed operator $H_k$. Under deep random circuits, matrix elements of fixed operators concentrate, and gradients decay exponentially. Despite this, many objectives are non-linear in measured statistics --- divergences \cite{kieferova_quantum_2021}, likelihoods \cite{coopmans_sample_2024, verdon_quantum_2019} , risk functionals \cite{barkoutsos_improving_2020}. Considerable effort has extended the linear proof machinery to cover specific non-linear losses, typically through bounded-sensitivity assumptions that allow gradient-level transfer from underlying linear expectations \cite{thanasilp_subtleties_2023}. These results are important but non-exhaustive: they establish inheritance for particular loss families under particular conditions, not a universal reduction. The broader question — for which objectives is the fixed-observable reduction structurally available, and what governs gradients when it is not — has remained largely implicit.

The structure of the paper is as follows. A fixed-observable representation exists if and only if the loss is affine in measured statistics, relative to a chosen measurement interface (Theorem~\ref{thm:escape}). Beyond the affine boundary, gradients are governed by a chain-rule decomposition into three factors: model responsivity, loss-side signal, and transmittance (Section~\ref{sec:geometry}). This framework induces a loss-class dichotomy. Bounded-gradient losses inherit suppression from the underlying Jacobian; losses whose feature-space gradients can grow without bound are not subject to the same inheritance and can, in principle, counteract it through the chain rule (Corollary~\ref{cor:loss-classes}). But non-linearity alone is not enough: when the measurement interface spans all $2^n$ outcomes, both loss classes fail --- the exponentially wide interface itself enforces suppression. Compressing the interface to polynomial width lifts the dimensional obstruction and brings the dichotomy into play. As a numerical demonstration, we showcase the capability of an amplification-capable objective on a charge-conserving quantum system, which produces materially different gradient behavior than affine and Lipschitz baselines, while model-side responsivity emerges as the binding constraint. This frames a representation-design question: under what conditions does the amplification regime become accessible for natural tasks?

%% file: Theorems/theorem_A_final_rev.tex

\section{The Structural Boundary}\label{sec:structural_boundary}

We formalize when the fixed-observable structure exists.

\paragraph{Setup.} Let \begin{equation*} F_j(\rho) := \tr(\rho\, O_j), \qquad j = 1,\dots,m, \end{equation*} and collect these expectations into \begin{equation*} F(\rho) := \big(F_1(\rho),\dots,F_m(\rho)\big) \in \mathbb{R}^m. \end{equation*} The map \(F\) encodes a design choice: it specifies which linear statistics of the quantum state the objective is allowed to depend on. Different choices of \(\{O_j\}\) expose different summaries of \(\rho\) to the classical loss. We refer to this chosen family of statistics as the \emph{measurement interface}. For a differentiable ansatz \(\rho(\theta)\) with parameters \(\theta \in \mathbb{R}^P\), define \begin{equation*} L(\theta) = f(F(\rho(\theta))). \end{equation*} We say that \(f\) is \emph{affine} on a region \(U \subset \mathbb{R}^m\) if there exist \(a \in \mathbb{R}^m\) and \(c \in \mathbb{R}\) such that \begin{equation*} f(F) = a^\top F + c \qquad \text{for all } F \in U. \end{equation*} Otherwise, \(f\) is \emph{non-affine} on \(U\). In the special case \(c=0\), affineness reduces to linearity. We use ``affine'' rather than ``linear'' because the fixed-observable template tolerates an additive constant without modification, and the distinction plays no structural role in what follows. Throughout, affineness and non-affineness are understood relative to the chosen measurement interface.


\begin{theorem}
\label{thm:escape}
Let $U \subset \mathbb{R}^m$ be a non-empty open set explored by the ansatz, i.e.\
\begin{equation*}
  U \;\subset\; \{F(\rho(\theta)) : \theta \in \mathbb{R}^P\}.
\end{equation*}
Then the following are equivalent:
\begin{enumerate}
\item
There exist $H \in \mathrm{span}\{O_1,\dots,O_m\}$ and $c \in \mathbb{R}$ such that
\begin{equation*}
  L(\theta) = \tr(H\,\rho(\theta)) + c
\end{equation*}
for all $\theta$ with $F(\rho(\theta)) \in U$.
\item
$f$ is affine on $U$, i.e.\ there exist $a \in \mathbb{R}^m$ and $c \in \mathbb{R}$
such that
\begin{equation*}
  f(F) = a^\top F + c \qquad \text{for all } F \in U.
\end{equation*}
\end{enumerate}
In particular, if $f$ is non-affine on $U$, then $L$ admits no fixed-observable
representation on $U$ relative to the chosen interface.
\end{theorem}

\begin{proof}
The direction (2)$\Rightarrow$(1) is immediate: define $H := \sum_j a_j O_j$ and
evaluate. We prove the contrapositive of (1)$\Rightarrow$(2).

Suppose, for contradiction, that $f$ is non-affine on $U$ and there exist
$H \in \mathrm{span}\{O_1,\dots,O_m\}$ and $c \in \mathbb{R}$ such that
$L(\theta) = \tr(H\,\rho(\theta)) + c$ for all $\theta$ with $F(\rho(\theta)) \in U$.
Write $H = \sum_{j=1}^m a_j O_j$ for some $a \in \mathbb{R}^m$.  Then for every such
$\theta$,
\begin{equation*}
  f(F(\rho(\theta)))
  \;=\; L(\theta)
  \;=\; \tr(H\,\rho(\theta)) + c
  \;=\; \sum_{j=1}^m a_j\,\tr(O_j\,\rho(\theta)) + c
  \;=\; a^\top F(\rho(\theta)) + c.
\end{equation*}
Because the ansatz explores $U$, every $F \in U$ is attained as $F(\rho(\theta))$
for some $\theta$.  Hence
\begin{equation*}
  f(F) = a^\top F + c \qquad \text{for all } F \in U,
\end{equation*}
so $f$ is affine on $U$, contradicting the assumption.
\end{proof}

\paragraph{Scope.}
We restrict throughout to observables fixed \emph{a priori}---the setting of concentration-based barren-plateau proofs.  Theorem~\ref{thm:escape} identifies the exact boundary of that template relative to the chosen measurement interface: fixed-observable structure exists exactly in the affine regime.  Non-affine losses lie on the other side of that boundary.

Allowing parameter-dependent observables moves to a different regime: such operators can be constructed post hoc to reproduce local gradients, but they do not provide structural counterexamples to concentration-based plateau theory (Appendix~\ref{app:iff}, Remark~\ref{rem:theta_dependent_obs}).

\paragraph{Remark.}
A convenient sufficient condition for the open-region hypothesis is that the feature Jacobian $J_F(\theta) := \partial F / \partial \theta(\theta)$ have full row rank on an open parameter region; then $F$ is locally open onto feature space, so the ansatz explores an open set of features. Appendix~\ref{app:iff} records the interface-relative formalization and related scope remarks.



\paragraph{Prior work and the structural boundary.}
Prior work around barren plateaus can be read through four interlinked layers. First, a large mitigation literature studies how concentration can be avoided \emph{within} the standard fixed-observable template, through shallow circuits, local observables, symmetry-preserving ans\"atze, and systematic initialization strategies~\cite{cerezo_cost_2021,uvarov_barren_2021,meyer_exploiting_2023,grant_initialization_2019, lerch_iqp_2026}. Second, a deeper concern has emerged about this program: many provably BP-free constructions appear to owe their trainability to restriction to classically manageable polynomial operator subspaces, effective shallowness, or other classically tame descriptions, while others exhibit superpolynomially few good local minima despite the absence of plateaus, calling into question whether provable BP-freedom inside the affine template delivers a quantum advantage at all~\cite{cerezo_does_2025,bermejo_quantum_2026,anschuetz_quantum_2022,mele_noise-induced_2026}. Third, work on non-linear losses divides into two strands. One strand derives conditional transfer results showing that barren-plateau scaling of underlying linear expectation values carries over to broad classes of non-linear losses under bounded-sensitivity assumptions; this is most explicit in Thanasilp et al., whose framework covers a broad separable class but does not characterize when such a transfer is structurally available in general~\cite{thanasilp_subtleties_2023}. The other strand provides empirical and theoretical evidence that particular non-linear objectives — Rényi-based losses, relative-entropy QBM training, and free-energy or thermal-state objectives — can train favorably in regimes where standard linear-template reasoning is incomplete or does not straightforwardly apply, but these remain objective-specific analyses rather than a general structural characterization~\cite{kieferova_quantum_2021,coopmans_sample_2024,verdon_quantum_2019}. Fourth, recent work identifies practical obstacles beyond the proof template itself: exact-form objectives can become uninformative under finite-shot estimation, and measured quantities can be concentrated in ways that no classical post-processing can rescue~\cite{rudolph_trainability_2024,saem_pitfalls_2025}. 

Between these strands sits a broader perspective from adjacent quantum machine learning work: trainability is necessary but not sufficient for usefulness. Recent work emphasizes that successful learning depends crucially on the inductive bias encoded by the parametrization, and that quantum learning models are often best understood at the level of their induced representations or feature spaces; in some cases, the resulting learned behavior may already admit classically tame surrogate descriptions~\cite{jerbi_quantum_2023,meyer_exploiting_2023,schreiber_classical_2023}. The unit of analysis thus shifts from the loss alone to the representation—or interface—through which the model is exposed to the loss.

Theorem~\ref{thm:escape} addresses the structural question that sits upstream of these layers: when does the objective itself admit a fixed-observable representation relative to a chosen measurement interface? It shows that this reduction is available if and only if the objective is affine on the explored region. In this sense, the theorem does not compete with the mitigation, simulability, inheritance, or estimation results above; it identifies the boundary that determines which of them are in play.

\paragraph{Beyond the boundary.}
In the affine regime, \(f(F)=a^\top F+c\), so \(\nabla_F f(F)=a\) is constant and
$\nabla_\theta L(\theta)=J_F(\theta)^\top a.$
Beyond the fixed-observable regime, this constant vector is replaced by the
state-dependent quantity \(\nabla_F f(F(\theta))\), so gradients are governed by
\begin{equation*}
  \nabla_\theta L(\theta)=J_F(\theta)^\top \nabla_F f(F(\theta)),
\end{equation*}
whose three factors---model responsivity, loss-side signal, and gradient
transmittance---are developed next.

%% file: Theorems/theorem_B_rev_final.tex
\section{Gradients in the Non-Linear Regime}\label{sec:geometry}
We now make the three factors precise. 
 
\subsection{The chain rule}
For a differentiable ansatz \(\rho(\theta)\) with parameters \(\theta\in\mathbb R^P\), write
\[
F(\theta):=F(\rho(\theta))\in\mathbb R^m
\]
for the features of the chosen observables, and
\[
J_F(\theta):=\frac{\partial F}{\partial \theta}(\theta)\in\mathbb R^{m\times P}
\]
for the Jacobian of the map \(\theta\mapsto F(\theta)\). For a loss of the form $L(\theta)=f(F(\theta))$, the gradient separates cleanly as
\begin{equation}
\label{eq:chain_rule}
    \nabla_\theta L(\theta)
    = J_F(\theta)^\top g_F(\theta),
    \qquad
    g_F(\theta):=\nabla_F f\big(F(\theta)\big).
\end{equation}
From this point onward, $\|\cdot\|$ denotes the Euclidean norm on vectors unless otherwise stated. 
\begin{lemma}[Chain-rule upper bound]
\label{lem:chain-rule-upper}
For any $\theta$,
\begin{equation*}
    \|\nabla_\theta L(\theta)\|
    = \|J_F(\theta)^\top g_F(\theta)\|
    \;\le\;
    \sigma_{\max}\!\big(J_F(\theta)\big)\,
    \|g_F(\theta)\|.
\end{equation*}
\end{lemma}

\begin{proof}
This is the Euclidean operator-norm bound:
\[
\|A x\| \le \sigma_{\max}(A)\,\|x\|
\]
for any matrix $A$. Applying this with $A=J_F(\theta)^\top$ and
$x=g_F(\theta)$, and using that $J_F(\theta)$ and $J_F(\theta)^\top$
have the same singular values, gives the claim.
\end{proof}

The Jacobian $J_F$ captures model responsivity; the vector $g_F$ captures the loss-side signal. The third factor---how effectively that signal aligns with directions the model can move---is made precise next.

\subsection{Transmittance}
Let $\sigma_{\max}(J_F(\theta))$ denote the largest singular value of
$J_F(\theta)$, and let $u_{\max}(\theta)$ be any corresponding unit left singular vector.
This $u_{\max}$ identifies the direction in which the feature map is most responsive to
changes in $\theta$. We define the \emph{gradient transmittance}

\begin{equation}
\label{eq:transmittance}
    \mathcal{T}(\theta)
    := \frac{|\langle g_F(\theta),\,u_{\max}(\theta)\rangle|}
            {\|g_F(\theta)\|}
    \in[0,1],
\end{equation}
with the convention $\mathcal{T}(\theta)=0$ if $g_F(\theta)=0$.
That is, \(\mathcal{T}\) measures the cosine overlap between the loss-side signal and the model’s most responsive feature-space direction. When \(\mathcal{T}\approx 1\), the loss-side signal is transmitted efficiently into parameter space. When \(\mathcal{T}\approx 0\), the model is responsive, but mostly in directions that do not align with the loss.

\begin{theorem}[Chain-rule amplification]
\label{thm:amplification}
For any $\theta$ with $g_F(\theta)\neq 0$,
\begin{equation*}
    \|\nabla_\theta L(\theta)\|
    = \|J_F(\theta)^\top g_F(\theta)\|
    \;\ge\;
    \sigma_{\max}\!\big(J_F(\theta)\big)\;
    \mathcal{T}(\theta)\;
    \|g_F(\theta)\|.
\end{equation*}
\end{theorem}
\begin{proof}
Let $J_F = U \Sigma V^\top$ be a singular value decomposition with singular values
$\sigma_1 \ge \sigma_2 \ge \cdots \ge 0$ and corresponding left singular vectors
$u_1,u_2,\dots$. Decompose $g_F$ along the leading singular direction:
\[
g_F = g_\parallel + g_\perp,
\qquad
g_\parallel = \langle g_F,u_1\rangle\,u_1,
\qquad
g_\perp \perp u_1.
\]
Then $\|g_\parallel\| = |\langle g_F,u_1\rangle|$.

Since $u_1$ is the leading left singular vector of $J_F$, the operator $J_F^\top$ acts on
$g_\parallel$ with gain $\sigma_1$, so
\[
\|J_F^\top g_\parallel\|
=
\sigma_1\,\|g_\parallel\|
=
\sigma_1\,|\langle g_F,u_1\rangle|.
\]
Moreover, $J_F^\top g_\parallel$ and $J_F^\top g_\perp$ lie in orthogonal right-singular
subspaces, hence
\[
\|J_F^\top g_F\|^2
=
\|J_F^\top g_\parallel\|^2
+
\|J_F^\top g_\perp\|^2
\;\ge\;
\|J_F^\top g_\parallel\|^2
=
\sigma_1^2\,|\langle g_F,u_1\rangle|^2.
\]
Taking square roots and multiplying and dividing by $\|g_F\|$ gives
\[
\|\nabla_\theta L(\theta)\|
=
\|J_F^\top g_F\|
\;\ge\;
\sigma_1\,|\langle g_F,u_1\rangle|
=
\sigma_{\max}(J_F)\,
\frac{|\langle g_F,u_1\rangle|}{\|g_F\|}\,
\|g_F\|
=
\sigma_{\max}(J_F)\,\mathcal{T}(\theta)\,\|g_F\|,
\]
where the last equality uses the definition of transmittance in
Eq.~\eqref{eq:transmittance}.
\end{proof}

In particular, the exponential decay of $\|\nabla_\theta L(\theta)\|$ requires at least one of the three factors to decay exponentially.

Standard barren-plateau results are typically phrased in terms of the variance of individual partial derivatives, $\operatorname{Var}_\theta[\partial_k L]$, over a random initialization ensemble, whereas our analysis is formulated pointwise in terms of the full gradient norm. The next corollary provides the bridge between these viewpoints.

\begin{corollary}[Aggregate variance bound]
\label{cor:variance-bridge}
Let $\theta$ be drawn from a parameter distribution $\mathcal D$, and define a covariance matrix
\[
\Sigma_L := \operatorname{Cov}_{\theta\sim\mathcal D}(\nabla_\theta L).
\]
Then
\[
\operatorname{Tr}\Sigma_L
\le
\mathbb E_{\theta\sim\mathcal D}\!\left[
    (\sigma_{\max}(J_F(\theta))\,\|g_F(\theta)\|)^2
\right].
\]
In particular, for every $k=1,\dots,P$,
\[
\operatorname{Var}_{\theta\sim\mathcal D}[\partial_k L]
\le
\operatorname{Tr}\Sigma_L.
\]
\end{corollary}

\begin{proof}
Let $X := \nabla_\theta L(\theta)$, where $\theta\sim\mathcal D$. Then
\[
\operatorname{Tr}\Sigma_L
=
\mathbb E\!\left[\|X-\mathbb E[X]\|^2\right]
=
\mathbb E\!\left[\|X\|^2\right]
-
\|\mathbb E[X]\|^2
\le
\mathbb E\!\left[\|X\|^2\right].
\]
Applying Lemma~\ref{lem:chain-rule-upper} pointwise gives
\[
\|X\|^2
=
\|\nabla_\theta L(\theta)\|^2
\le
(\sigma_{\max}(J_F(\theta))\,\|g_F(\theta)\|)^2.
\]
Taking expectations yields
\[
\operatorname{Tr}\Sigma_L
\le
\mathbb E_{\theta\sim\mathcal D}\!\left[
    (\sigma_{\max}(J_F(\theta))\,\|g_F(\theta)\|)^2
\right].
\]
Finally, for each $k=1,\dots,P$,
\[
\operatorname{Var}_{\theta\sim\mathcal D}[\partial_k L]
=
(\Sigma_L)_{kk}
\le
\operatorname{Tr}\Sigma_L,
\]
since $\operatorname{Tr}\Sigma_L$ is the sum of the nonnegative diagonal entries of the covariance matrix $\Sigma_L$.
\end{proof}

\begin{remark}[ Variance-level interpretation]
\label{rem:norm-variance-bridge}
Corollary~\ref{cor:variance-bridge} translates norm-level upper bounds into control of the
aggregate gradient variance
\[
\operatorname{Tr}\Sigma_L
=
\sum_{k=1}^P \operatorname{Var}_{\theta\sim\mathcal D}[\partial_k L],
\]
thereby linking our norm-level analysis to the standard variance-of-partial-derivatives diagnostics used in the barren-plateau literature. The right-hand side,
\[
\mathbb E_{\theta\sim\mathcal D}\!\left[
    (\sigma_{\max}(J_F(\theta))\,\|g_F(\theta)\|)^2
\right],
\]
is the second moment of the pointwise chain-rule upper bound from Lemma~\ref{lem:chain-rule-upper}: it combines model responsivity through $\sigma_{\max}(J_F)$ and loss-side signal through $\|g_F\|$. No transmittance term appears because this is an upper bound: only the largest possible local amplification of the Jacobian enters.

The per-parameter conclusion is intentionally aggregate-to-component: it controls each individual variance through the total variance mass, but does not resolve how that variance is distributed across parameters without further structural assumptions. For the present purpose, however, no converse is needed: if the chain-rule upper bound is exponentially suppressed at random initialization, then the aggregate variance—and hence each individual partial-derivative variance—is exponentially suppressed as well.
\end{remark}

\subsection{Inheriting objectives and an amplification mechanism}
\label{subsec:two-loss-class}
Trainability in the non-affine regime is governed by three factors,
\begin{equation*}
    \text{model responsivity } \sigma_{\max}(J_F), \qquad
    \text{loss-side signal } \|g_F\|, \qquad
    \text{gradient transmittance } \mathcal{T}(\theta).
\end{equation*}
In the affine regime, $g_F$ is constant so the loss contributes no additional state-dependent weighting beyond the Jacobian. For non-affine losses, $g_F$ varies with the features, and the chain rule combines the two effects multiplicatively---so depending on the loss, Jacobian flattening can be either inherited or counteracted.

\begin{corollary}[Smooth/Lipschitz losses inherit exponential gradient suppression]
\label{cor:loss-classes}
Let $L(\theta)=f(F(\theta))$ with $J_F(\theta)$ and $g_F(\theta)$ as above. Assume $f$ is $C^1$ on an open set containing the reachable feature set $\{F(\theta):\theta\in\mathbb{R}^P\}$, and there exists a constant $L_f>0$ (independent of system size) such that
\begin{equation*}
    \|\nabla_F f(F)\| \le L_f
    \qquad
    \text{for all } F \in \{F(\theta):\theta\in\mathbb{R}^P\}.
\end{equation*}
If, under random initialization, the feature-map responsivity is exponentially suppressed in the sense that there exist $\alpha>0$ and $C_J>0$ such that
\[
\mathbb{E}\!\left[\sigma_{\max}(J_F(\theta))\right]\le C_J\,2^{-\alpha n},
\]
then
\begin{equation*}
    \mathbb{E}\big[\|\nabla_\theta L(\theta)\|\big]
    \;\le\;
    L_f\,\mathbb{E}\!\left[\sigma_{\max}(J_F(\theta))\right]
    \;\in\;
    \mathcal{O}(2^{-\alpha n}),
\end{equation*}
so $L$ inherits the same exponential gradient suppression.
\end{corollary}

This recovers the bounded-sensitivity transfer result identified in Thanasilp et al.~\cite{thanasilp_subtleties_2023}, now expressed within the present chain-rule framework so as to contrast it with the structurally different beyond-affine case analyzed next.

\begin{proposition}[Negative log-likelihood loss amplification]
\label{prop:nll-amplification}
Let $\mathsf{X}$ be a finite outcome set with $|\mathsf{X}| = N = 2^n$, and let
$p_\theta \in \Delta(\mathsf{X})$ denote the Born probabilities of a
parametrized pure state in some reference basis. Consider the negative
log-likelihood (NLL) loss
\begin{equation*}
    f_{\mathrm{NLL}}(p) = -\sum_x q(x) \log p(x).
\end{equation*}
Assume that, at random initialization in an expressive regime:
\begin{enumerate}
    \item $p_\theta(x) \approx 1/N$ for all $x \in \mathsf{X}$;
    \item the target distribution $q$ is approximately uniform on a support
    of size $s$;
    \item the transmittance $\mathcal{T}(\theta)$ does not vanish
    exponentially in $n$.
\end{enumerate}
Then the feature-space gradient has magnitude
\begin{equation*}
    \|\nabla_p f_{\mathrm{NLL}}(p_\theta)\| \approx \frac{N}{\sqrt{s}},
\end{equation*}
so that for high-entropy data ($s \approx N$), the loss-side signal grows as
$2^{n/2}$, and for sharply concentrated data ($s = \mathcal{O}(1)$), it grows
as $2^n$. Combined with Theorem~\ref{thm:amplification}, this implies that
the parameter-space gradient $\|\nabla_\theta L\|$ need not inherit the
exponential decay of $\sigma_{\max}(J_F)$: the loss-side growth can, in
principle, compensate for Jacobian flattening.
\end{proposition}

\begin{remark}[ Epistemic scope]
Proposition~\ref{prop:nll-amplification} is an analytic (oracle-model)
statement about the loss-side signal $\|g_F\|$. It does not address
finite-shot estimation, where the same non-Lipschitz structure inflates
estimator variance and can make the loss effectively uninformative even
when the analytic gradient is large. The interaction between the
amplification mechanism and shot noise is taken up in
Section~\ref{sec:loss_function_duality} and Section~\ref{sec:compressed_interface}.
\end{remark}

%% file: Theorems/corollary_C.tex
\subsection{The null model for transmittance}

The lower bound of Theorem~\ref{thm:amplification} becomes quantitatively informative only if $\mathcal{T}(\theta)$ does not collapse. We adopt an isotropic null model for $\mathcal{T}(\theta)$ on unstructured interfaces: the loss-side signal and the dominant feature-response direction are independent and uniformly oriented. This gives a baseline scale against which structured compressed interfaces can be compared.
\medskip

\begin{lemma}[Null model for gradient transmittance]
Assume that under random initialization of $\theta$:
\begin{enumerate}
    \item the leading left singular vector of $J_F(\theta)$ is uniformly distributed on the sphere $S^{m-1}$, independently of its singular values;
    \item the normalized loss gradient $\hat g_F(\theta) := g_F(\theta)/\|g_F(\theta)\|$ points in a uniformly random direction in $\mathbb{R}^m$, independently of $J_F(\theta)$.
\end{enumerate}
Then $\mathcal{T}(\theta)$ has the same distribution as $|\langle u,v\rangle|$, where
$u,v$ are independent, uniformly random unit vectors in $\mathbb{R}^m$. In particular,
\[
    \mathbb{E}[\mathcal{T}(\theta)] = \Theta(1/\sqrt{m}),
\]
with concentration at this scale.
\end{lemma}

\begin{proof}[Proof sketch]
By assumption, $u_{\max}(J_F(\theta))$ is uniformly distributed on the sphere $S^{m-1}$.
Assumption (ii) states that $\hat g_F(\theta)$ is also uniformly distributed on $S^{m-1}$,
independently of $u_{\max}(J_F(\theta))$. Therefore
\[
    \mathcal{T}(\theta)
    =
    |\langle \hat g_F(\theta),\,u_{\max}(J_F(\theta))\rangle|
\]
has the same distribution as $|\langle u,v\rangle|$, where $u$ and $v$ are independent,
uniformly random unit vectors on $S^{m-1}$. For such $u,v$, one has
\[
\mathbb E\,|\langle u,v\rangle|=\Theta(1/\sqrt m),
\]
with concentration on the same scale; see, e.g., \cite{vershynin_high-dimensional_2018}.
\end{proof}

Under the isotropic baseline, Theorem~\ref{thm:amplification} becomes quantitatively informative: chain-rule amplification remains available when $\|g_F(\theta)\|$ grows faster than $\sqrt{m}$ relative to the decay of $\sigma_{\max}(J_F(\theta))$. This is a necessary scaling check against the unstructured baseline; whether it is actually met on a given interface depends on the loss class, the feature dimension, and the choice of measurement interface. We begin with the hardest case.



%% file: 04-qml-case-study-final-v2.tex
\section{The Exponentially Wide Interface}
\label{sec:loss_function_duality}

The chain-rule framework identifies a genuine amplification mechanism, but on the exponentially wide interface neither loss class yields a trainable regime — amplification-capable losses through the neutralization of their loss-side signal, inheriting losses through their inherent Lipschitz structure. This section is a post-mortem: the beyond-the-boundary mechanism is present in principle but unusable in practice, and this motivates the search for a different interface rather than a different objective alone.

Let $\mathsf{X}=\{0,1\}^n$ be the outcome space of computational-basis measurements, let $p_\theta(x)=|\langle x|\psi(\theta)\rangle|^2$ be the model distribution, and let $q_{\mathrm{data}}(x)$ denote the target. Take the feature map to be the full probability vector
\begin{equation*}
    F(\rho(\theta)) = p_\theta \in \Delta(\mathsf{X}) \subset \mathbb{R}^{2^n},
\end{equation*}
so the feature dimension is $m=2^n$. Learning an unconstrained distribution over $2^n$ outcomes to total variation $\varepsilon$ requires $\Omega(2^n/\varepsilon^2)$ samples \cite{canonne_survey_2020}. Within the three-factor decomposition, the two loss classes fail through different combinations of mechanisms, which we examine in turn.

\subsection{Amplification-capable: maximum likelihood (forward KL)}

Maximum-likelihood training minimizes the negative log-likelihood
\begin{equation*}
    L_{\mathrm{NLL}}(\theta)
    = -\sum_x q_{\mathrm{data}}(x)\,\log p_\theta(x)
    = f_{\mathrm{NLL}}(p_\theta).
\end{equation*}
As discussed in Proposition~\ref{prop:nll-amplification}, $f_{\mathrm{NLL}}$ is analytically
amplification-capable: its feature-space gradient has components
$-q_{\mathrm{data}}(x)/p_\theta(x)$, which diverge as $p_\theta(x)\to 0$, and in a near-uniform
expressive regime one has $\|g_F\|\approx N/\sqrt{s}$ with $N=2^n$ and
$s=|\mathrm{supp}(q_{\mathrm{data}})|$. In this sense, the chain-rule geometry identifies a genuine
loss-side mechanism by which non-Lipschitz objectives can counteract Jacobian flattening.

The exponentially wide interface is, however, hostile to this mechanism. Three independent features
of the regime neutralize its potential gains. First, in practical training one clips or tempers the
loss—replacing $p_\theta(x)$ by $\max\{p_\theta(x),\varepsilon_n\}$ or an equivalent smoothed
variant—to control estimator variance. Under such clipping, $\|\nabla_p f_{\mathrm{NLL}}\|$ is
bounded by $\mathrm{poly}(n,1/\varepsilon_n)$; once $\varepsilon_n$ is at least
inverse-polynomial, as in the bounded-sensitivity assumptions used in
Thanasilp et al.~\cite{thanasilp_subtleties_2023}, NLL is pushed back toward the
Lipschitz/inheriting regime. Second, in the expressive random-initialization regime
$p_\theta(x)\approx 1/N$, the model begins near a broad, nearly uniform distribution over the
exponentially many outcomes, so the rare-event structure on which amplification would have to act is
not yet operationally localized. Third, the interface itself carries an exponential width penalty:
with $m=2^n$ exposed probabilities, the isotropic transmittance baseline scales as
$\mathcal{T}\sim 2^{-n/2}$, suppressing the chain-rule lower bound even when the loss-side signal is large.

This is consistent with recent generative-modeling results: exact explicit objectives and their finite-shot estimators can behave very differently, and in implicit generative models polynomial-shot estimates of explicit losses can become effectively uninformative~\cite{rudolph_trainability_2024}; more generally, concentration at the level of the measured quantities cannot be rescued by classical post-processing~\cite{saem_pitfalls_2025}. The loss-side mechanism is real, but on the exponentially wide interface it is neutralized before it can become operationally useful.

\subsection{Structurally inheriting: adversarial training (reverse KL / JSD)}

In adversarial training, a generator is paired with a discriminator, and with an optimal
discriminator the generator minimizes a divergence closely related to the reverse KL or
Jensen--Shannon (JSD) divergence between $p_\theta$ and $q_{\mathrm{data}}$. A convenient
proxy is the reverse KL
\begin{equation*}
    D_{\mathrm{KL}}(p_\theta\|q_{\mathrm{data}})
    = \sum_x p_\theta(x) \log\!\frac{p_\theta(x)}{q_{\mathrm{data}}(x)}.
\end{equation*}
Differentiating with respect to $p_\theta$ gives
\begin{equation*}
    \frac{\partial}{\partial p(x)} D_{\mathrm{KL}}(p\|q_{\mathrm{data}})
    = \log\!\frac{p(x)}{q_{\mathrm{data}}(x)} + 1.
\end{equation*}
In the near-uniform regime $p_\theta(x)\approx 1/N$ and for data distributions
$q_{\mathrm{data}}$ that are not exponentially spiky, the log term grows at most
polylogarithmically in $N$; in particular $\|\nabla_p D_{\mathrm{KL}}(p\|q_{\mathrm{data}})\|$
is at most polynomial in $n$. The same holds for the JSD generator losses used in qGANs,
which differ from reverse KL only by a smooth reweighting of terms.

Thus reverse-KL/JSD sit firmly in the \emph{Lipschitz/inheriting} class of
Corollary~\ref{cor:loss-classes}: their feature-space gradients are uniformly bounded (up to poly factors), so any exponential decay in the feature-map Jacobian $J_F$ is inherited by the parameter gradients. Unlike NLL, there is no non-Lipschitz amplification to exploit even in the analytic/oracle limit. 

\medskip
The distinction between these two failure modes is academic at $m=2^n$. The question is which remains concrete when the interface itself is compressed.


%% file: 05-compressed-interface-v2.tex
\section{Compressed Feature Maps, Classical Heads, and Representation}
\label{sec:compressed_interface}

The exponentially wide interface forecloses both loss classes. Compression changes the question but not necessarily the answer. A viable interface must satisfy three conditions: its width must be polynomial in $n$; its observables must be resolvable at polynomial shots; and its exposed statistics must be sufficient for the task. Whether any interface satisfies all three simultaneously for a given objective is an open question. This is a representational design question: which statistics the quantum system exposes, and how the loss processes them, together determine whether the three factors of Section \ref{sec:geometry} can be jointly controlled. This section formalizes the design space; the hypothesis that natural tasks admit solutions is taken up in Section \ref{sec:discussion}.

\subsection{Compressed feature maps and classical heads}
\label{subsec:composite_head}

Consider a feature map
\begin{equation*}
    F:\mathcal{D}(\mathcal{H})\to\mathbb{R}^m,\qquad
    \rho\mapsto F(\rho),
\end{equation*}
of dimension $m\ll 2^n$. The coordinates of $F$ can be expectation values of selected observables, marginals or correlators, binned probabilities, or other aggregated, coarse-grained statistics---each a linear combination of Pauli expectations, so the choice of $F$ amounts to selecting which directions in the $4^n$-dimensional operator space the model exposes to the loss. The feature dimension $m$ is the key design parameter: it determines both the size of the measurement interface and the baseline scale of gradient transmittance ($\mathbb{E}[|\mathcal{T}|]\sim 1/\sqrt{m}$ at random orientation).

On top of $F$ we allow a \emph{classical head} that processes the exposed statistics before producing a scalar loss. The head splits naturally into two roles: a preprocessing stage $T$ that may reshape the exposed statistics — through normalization, coarse-graining, or a learned classical map — and a scalar objective $C$ that scores the result. This split makes the framework cover structures that a single $f$ would obscure: qGAN discriminators and hybrid quantum-classical heads fit the same template. Concretely, we write
\begin{equation}
\label{eq:compressed_head}
L(\theta) = C\big(T(F(\rho(\theta)))\big),
\end{equation}
where $T:\mathbb{R}^m\to\mathbb{R}^{m'}$ is a (possibly non-linear) transformation and
$C:\mathbb{R}^{m'}\to\mathbb{R}$ is a scalar objective. The composite map $f := C\circ T$
plays the role of the loss in feature space. By the chain rule,
\begin{equation*}
    \nabla_\theta L(\theta)
    = J_F(\theta)^\top g_F^{\mathrm{eff}}(\theta),
\end{equation*}
where $J_F(\theta)=\partial F/\partial\theta(\theta)\in\mathbb{R}^{m\times P}$ is the feature
Jacobian and
\begin{equation}
\label{eq:compressed_head_effgrad}
    g_F^{\mathrm{eff}}(F)
    := \Big(\frac{\partial T}{\partial F}(F)\Big)^\top \nabla C\big(T(F)\big)
\end{equation}
is the \emph{effective feature-space gradient} induced by the classical head. All results of
Section~\ref{sec:geometry} apply with $g_F$ replaced by $g_F^{\mathrm{eff}}$: in
particular the chain-rule amplification bound becomes
\begin{equation*}
    \|\nabla_\theta L(\theta)\|
    \;\ge\;
    \sigma_{\max}(J_F(\theta))\,
    \mathcal{T}(\theta)\,
    \big\|g_F^{\mathrm{eff}}(\theta)\big\|.
\end{equation*}
Thus, in the compressed setting, trainability is again governed by the three factors of Section~\ref{sec:geometry}---responsivity $\sigma_{\max}(J_F)$, effective loss-side signal $\|g_F^{\mathrm{eff}}\|$, and transmittance $\mathcal{T}$---but the feature dimension $m$ now enters as a design parameter. Reducing $m$ from $2^n$ to $\mathrm{poly}(n)$ removes the dimensional source of exponentially small transmittance. This does not guarantee favorable scaling. The isotropic null model is an initialization-time baseline; structured compressed interfaces need not remain isotropic. Refined anisotropic baselines, controlled by an effective spectral dimension rather than the raw $m$, are developed in Appendix~\ref{app:elliptical_refinements}.

\subsection{The Goldilocks question}


Compression is not merely width reduction; it is a representation choice. The feature map $F$ determines which directions in operator space the model exposes to the loss, and the classical head determines how those exposed statistics are scored. Trainability depends on responsivity, effective loss-side signal, and transmittance — all three shaped by the choice of representation, not by width alone.

This is a Goldilocks question. If the interface is too coarse, it no longer exposes statistics sufficient for the task. If it is too fine, the exponential width obstruction returns. If it is too hard to estimate, the interface is unusable even when the chain-rule geometry is otherwise favorable. Prior work has identified related failure modes in different forms: overly restricted representations may collapse into classically tame operator subspaces or admit learned classical surrogates; overly fine interfaces reintroduce exponential-width concentration; and even favorable interfaces can fail through measurement-level concentration that no classical post-processing can rescue~\cite{schreiber_classical_2023,cerezo_does_2025,rudolph_trainability_2024,saem_pitfalls_2025}. Whether the complement of these failure regions is nonempty for natural tasks is the hypothesis taken up in Section~\ref{sec:discussion}. Section~\ref{sec:hydro_case_study} probes one candidate: a joint-block interface on charge-conserving local dynamics, chosen to test whether the three factors can be jointly controlled in a physically motivated setting.

%% file: 06-hydro-case-study-rev2.tex
\section{Numerical Demonstration on a Compressed Non-Affine Interface}
\label{sec:hydro_case_study}

In this section we study a concrete compressed, non-affine learning task built from charge-conserving local dynamics. The relevant structure is naturally expressed through coarse density profiles and their correlations rather than full bitstring distributions, making the setting a useful testbed for a polynomial-width interface.

\subsection{Joint-block interface for charge-conserving dynamics}
\label{subsec:hydro_setup}

Concretely, we study a local $U(1)$-conserving circuit initialized in a domain-wall state, which provides a spatially inhomogeneous initial condition, and define the task on the joint distribution of block Hamming weights across $b$ contiguous regions,
\begin{equation*}
    F(\rho) = p_\rho(w_1,\dots,w_b).
\end{equation*}

The polynomial width of this interface follows directly from the block structure: for fixed $b$,
\begin{equation*}
    m(n) = \prod_{j=1}^{b}(|B_j|+1)
    \sim \left(\frac{n}{b}+1\right)^{b},
\end{equation*}
so the feature dimension avoids the exponentially wide regime in which $m \sim 2^n$.

This interface lies between two extremes. It is richer than a small collection of global observables, since it retains coarse spatial structure and cross-block correlations, but it is substantially more compressed than full microstate learning. In particular, it is well suited to the charge-conserving setting because the relevant structure is carried by coarse spatial fluctuations rather than by exact bitstring configurations.

\begin{figure}
    \centering
    \includegraphics[width=\textwidth]{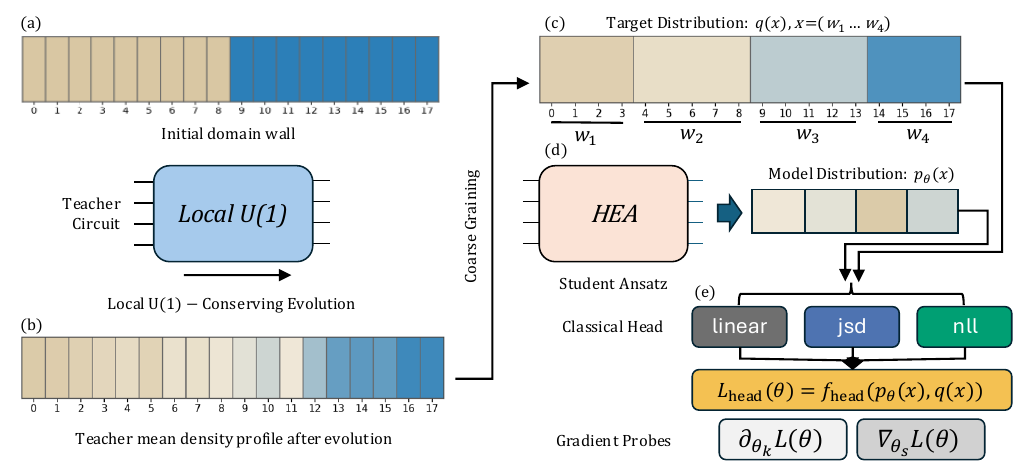}
    \caption{
    Setup of the numerical demonstration. (a) A domain-wall initial state on $n$ qubits (shown here for $n=18$). (b) A local $U(1)$-conserving circuit evolves the state, producing a diffusively broadened mean density profile. (c) The evolved teacher state is mapped to the joint block-weight target distribution $q(x)$, with $x=(w_1,\ldots,w_4)$, over $b=4$ contiguous blocks. (d) An unconstrained hardware-efficient ansatz produces the student distribution $p_\theta(x)$ on the same interface. (e) Three classical heads---linear (affine), JSD (inheriting), and NLL (amplification-capable)---define competing objectives $L(\theta)=f_{\mathrm{head}}(p_\theta(x),q(x))$, whose gradients are probed in single-parameter $(\partial_{\theta_k}L)$ and multi-parameter subspace $(\nabla_{\theta_S}L)$ settings. In the numerics, $q(x)$ is represented by a fixed finite-shot estimate; see Appendix~\ref{app:frontier_defs}.
    }
    \label{fig:hydro_numerics}
\end{figure}

At the same time, polynomial width alone does not guarantee polynomial responsivity. Even when $m(n)=\mathrm{poly}(n)$ for fixed $b$, each block variable still aggregates over $|B_j|=\Theta(n/b)$ qubits, so the exposed observables become increasingly extensive with system size. The central numerical question is therefore not merely whether compression reduces the feature dimension, but whether it does so while leaving enough model-side responsivity and alignment for non-affine amplification to matter.

The teacher is a local $U(1)$-conserving circuit evolved from a domain-wall state, with depth $d_T(n)=\left\lceil (n/4)^2 \right\rceil$ taken as a heuristic scaling choice to keep the task in a roughly comparable intermediate-relaxation regime across system sizes. The prefactor $1/4$ is not optimized. The student is an expressive hardware-efficient ansatz of the same size; we deliberately leave the student unconstrained by $U(1)$ symmetry, since imposing it would inject favorable structure into the variational family and obscure the role of the head and interface in shaping gradients. We compare three classical heads applied to the joint block-weight distribution $p = p_\rho(w_1,\dots,w_b)$ with teacher target $q$:
\begin{equation}
\label{eq:hydro_heads}
    L_{\mathrm{linear}} = -\sum_x q_x\,p_x,
    \qquad
    L_{\mathrm{JSD}} = \tfrac{1}{2}D_{\mathrm{KL}}(p\|m)
        + \tfrac{1}{2}D_{\mathrm{KL}}(q\|m),
    \qquad
    L_{\mathrm{NLL}} = -\sum_x q_x\,\log p_x,
\end{equation}
where $m = (p+q)/2$ and all probabilities are lightly smoothed.\footnote{Each distribution is additively smoothed with $\varepsilon = 10^{-12}$ and renormalized; see Appendix~\ref{app:frontier_defs} for details.} Their feature-space gradients are $\nabla_{p_x}L_{\mathrm{linear}} = -q_x$ (affine), $\nabla_{p_x}L_{\mathrm{JSD}} = \tfrac{1}{2}\log(p_x/m_x)$ (inheriting), and $\nabla_{p_x}L_{\mathrm{NLL}} = -q_x/p_x$ (amplification-capable) — the loss-classes of Corollary~\ref{cor:loss-classes}, Proposition~\ref{prop:nll-amplification} made concrete in this setup.

All finite-shot results are reported at accepted frontiers chosen so that the finite-shot gradient estimates track the corresponding exact-gradient diagnostics at the ensemble level. We summarize resolution by the median circuit-level signal-to-noise ratio $\mathrm{MedSNR}$, and, for the multi-parameter probe, fidelity by the median relative bias $\mathrm{MedRelBias}$ against the exact subspace gradient. The single-parameter frontier is the smallest shot budget with $\mathrm{MedSNR}\ge 2$; the multi-parameter frontier additionally requires $\mathrm{MedRelBias}\le 0.5$. Formal definitions and threshold choices are given in Appendix~\ref{app:frontier_defs}.

\subsection{Single-parameter finite-shot probe}
\label{subsec:hydro_single_param}

We first study a single parameter direction under finite-shot estimation, in direct correspondence with the standard barren-plateau diagnostic. Holding the circuit family and measurement interface fixed, we compare the three heads defined above.

At matched shot budgets, the non-affine NLL head produces substantially
larger resolved gradients than the affine baseline and JSD
(Fig.~\ref{fig:hydro_singleparam_v21}). At $n=24$, NLL resolved gradients
are approximately $10^4$ times larger than the linear baseline at matched
shot cost. Over the explored system-size window, the linear and JSD curves
are strongly consistent with exponential-like decay, whereas NLL decays
markedly more slowly, with the scaling class quantified in
Section~\ref{subsec:scaling_classification}.

\begin{figure*}[h!]
    \centering
    \includegraphics[width=\textwidth]{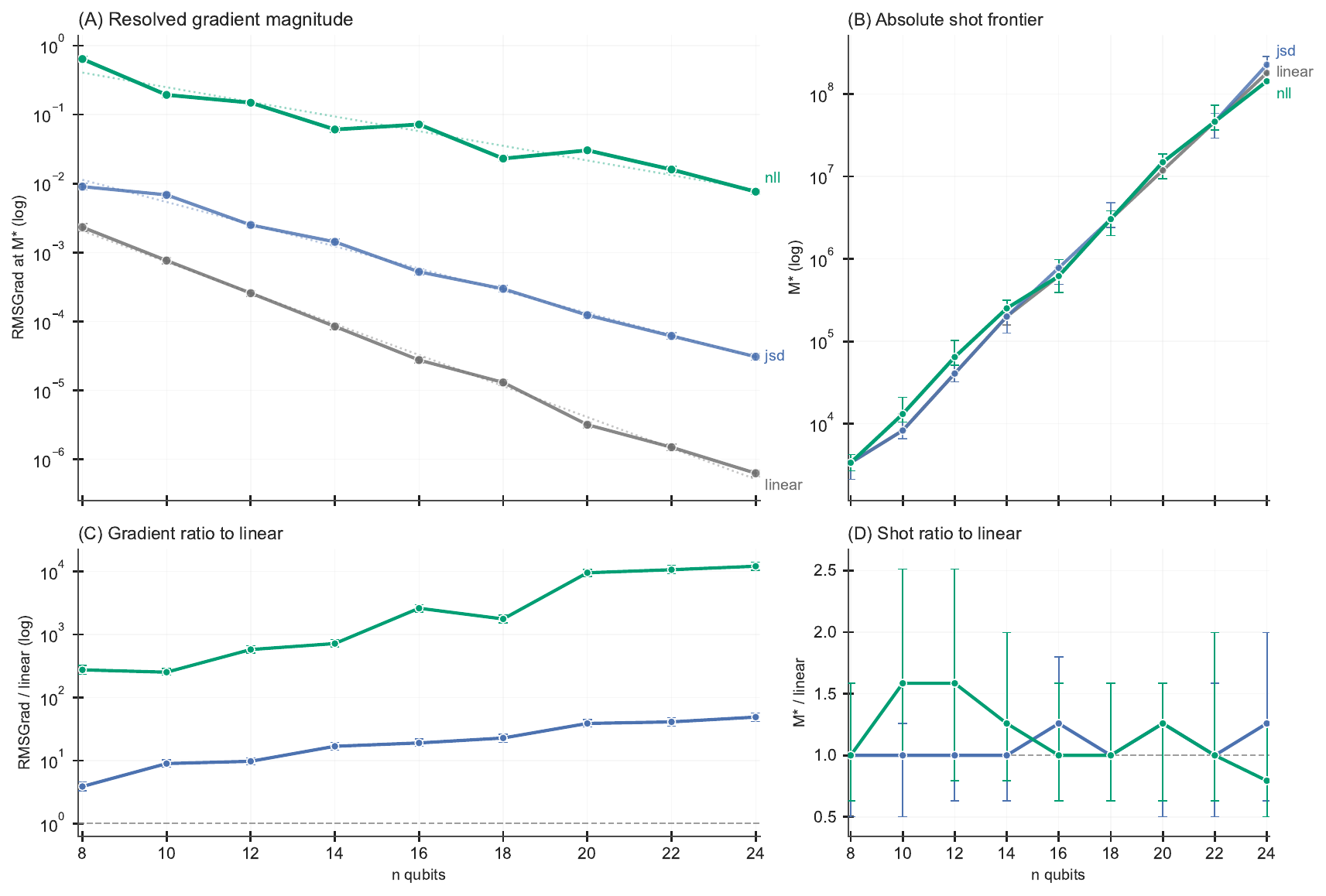}
    \caption{At fixed block-interface width $b=4$, the non-affine NLL head produces resolved gradients roughly $10^4\times$ larger than the linear baseline at $n=24$, while the required shot budgets remain in the same broad exponential-like regime up to a moderate constant-factor overhead. Panel~(A): resolved gradient magnitude at the accepted $\mathrm{MedSNR}\ge 2$ frontier. Panel~(B): absolute shot frontier $M^*$. Panels~(C) and~(D): gradient and shot ratios relative to the linear baseline. Thus the classical head materially changes the resolved training gradient scale without changing the dominant shot-budget scaling class.}
    \label{fig:hydro_singleparam_v21}
\end{figure*}

\subsection{Random coordinate-subset sketch}
\label{subsec:hydro_subspace}

The single-parameter probe isolates the loss-side mechanism but does not capture how gradients are used in optimization, where multiple coordinates move simultaneously. To address this, we introduce a random $s=32$ coordinate-subset sketch: for each circuit instance we evaluate exact parameter-shift gradients on a fixed subset $S$ of parameters, giving a tractable proxy for local multi-parameter trainability.

Within this sketch we estimate the chain-rule factors of
Section~\ref{sec:geometry}: the feature-space gradient norm $\|g_F\|$, the
leading feature-Jacobian scale $\sigma_{\max}(J_{F,S})$, the transmittance
$\mathcal{T}_S$, and the transmitted subspace gradient norm
$\|J_{F,S}^\top g_F\|$
(Fig.~\ref{fig:hydro_chainrule_theorem}). The numerical hierarchy is
consistent: NLL produces the largest feature-space signal, the transmitted
gradient norm is ordered $\mathrm{NLL}\gg\mathrm{JSD}\gg\mathrm{linear}$, and the shot frontiers for this probe (Appendix~\ref{app:frontier_defs}) are larger than the single-parameter frontiers but consistent with the same exponential scaling class. The dominant remaining bottleneck is the model-side responsivity $\sigma_{\max}(J_{F,S})$, which decays strongly and is largely head-independent.

\begin{figure*}[t]
    \centering
    \includegraphics[width=\textwidth]{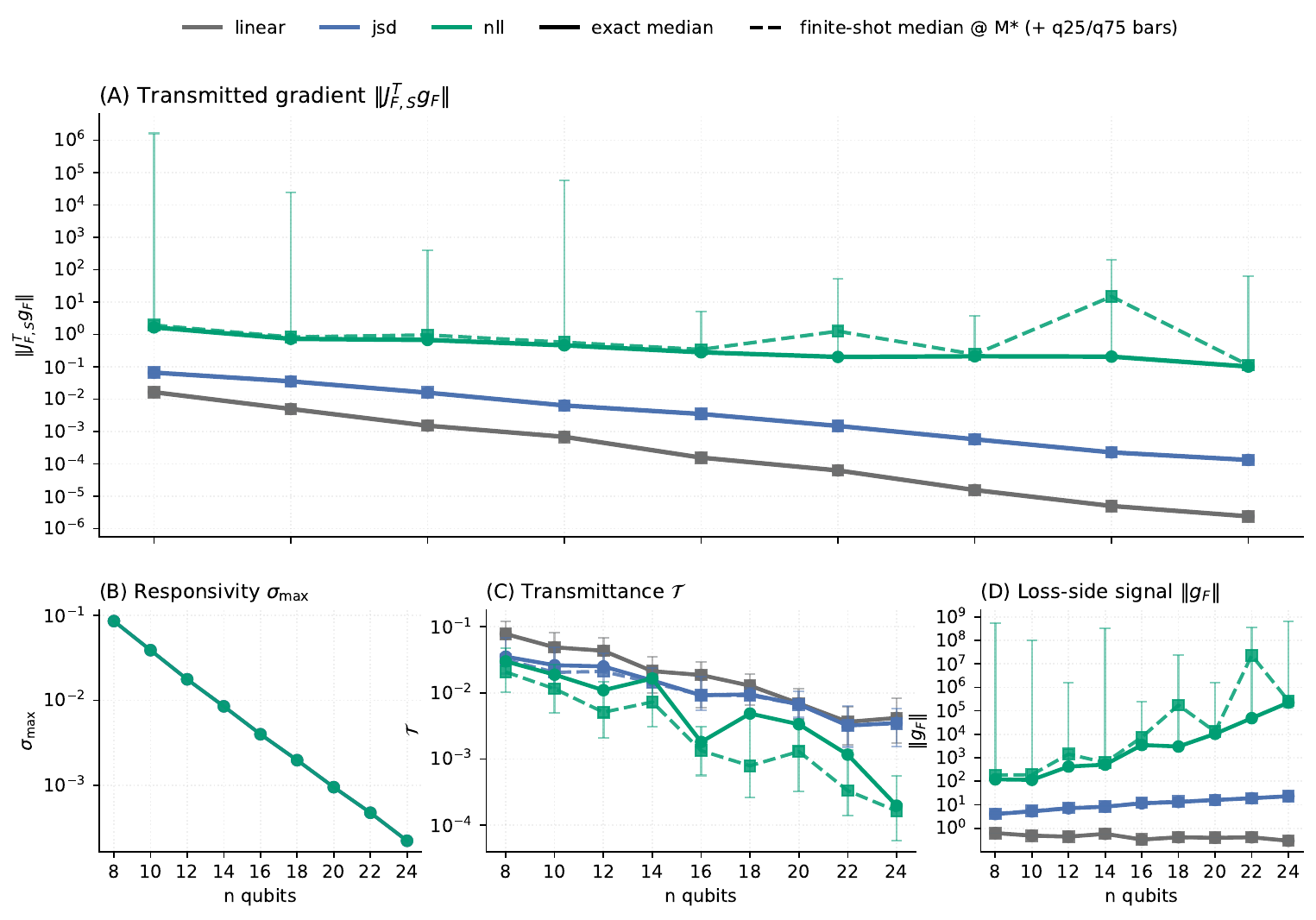}
    \caption{In the chain-rule decomposition of Theorem~\ref{thm:amplification}, the non-affine NLL head produces the largest loss-side signal $\|g_F\|$, the transmitted gradient norm orders as $\mathrm{NLL}\gg\mathrm{JSD}\gg\mathrm{linear}$, and model-side responsivity $\sigma_{\max}(J_{F,S})$ remains the dominant bottleneck across all heads. Panel~(A): transmitted gradient norm $\|J_{F,S}^\top g_F\|$. Panels~(B)--(D): responsivity $\sigma_{\max}(J_{F,S})$, transmittance $\mathcal{T}$, and loss-side signal $\|g_F\|$. Solid lines are exact medians across a 60-circuit ensemble; dashed markers are finite-shot medians at the accepted frontier with $q25/q75$ bars. The predicted hierarchy holds both exactly and under finite-shot estimation, though NLL exhibits broader cross-circuit spread (Appendix~\ref{app:ridgeline}).}
    \label{fig:hydro_chainrule_theorem}
\end{figure*}

Across the decomposition, linear and JSD finite-shot medians track their exact counterparts closely. NLL retains a clear transmitted-gradient
advantage under finite-shot estimation, though with broader cross-circuit spread---reflecting the bimodal ensemble structure visible in Appendix~\ref{app:ridgeline}.

\subsection{Scaling classification}
\label{subsec:scaling_classification}

The mechanistic decomposition explains \emph{why} NLL produces larger
transmitted gradients than linear and JSD on the same interface. We now ask
whether that separation remains within the same exponential scaling class, or
whether the NLL trend departs from it.

\paragraph{Model families.}
We fit four equal-complexity, two-parameter models to the exact transmitted subspace gradient norm as a function of qubit count $n$. Each model takes the form $\log y = \beta_0 + \beta_1\phi(n)$ with a different feature map
$\phi$. These span the standard hierarchy from polynomial through quasi-polynomial to exponential decay. Because all four models have the same number of free parameters, corrected Akaike information criterion (AICc) reduces to a ranking of log-space fit quality up to a common penalty term.

\paragraph{Results.}
Table~\ref{tab:hydro_scaling_aicc} reports $\Delta$AICc relative to the
best-fitting model in each row for the exact multi-parameter
transmitted-gradient trend over $n\in\{8,10,\ldots,24\}$.

\begin{table}[t]
\centering
\renewcommand{\arraystretch}{1.2}
\input{figures/hydro_complexity_tables/multiparam_gradient_aicc_full.tex}
\caption{$\Delta$AICc values for four equal-complexity two-parameter scaling models fitted to the exact transmitted subspace gradient norm $\|J_{F,S}^\top g_F\|$ as a function of qubit count. Linear and JSD decisively favor the exponential class; the NLL trend is not well described by the exponential class over the accessible window but does not identify a unique non-exponential law. Entries are $\Delta$AICc relative to the row-wise winner (shown in bold), so lower is better; $\mathrm{AICc}_{\mathrm{best}}$ gives the absolute AICc value of that winner.}
\label{tab:hydro_scaling_aicc}
\end{table}

For linear and JSD, the exponential fit is isolated from all alternatives by
large $\Delta$AICc gaps---fully consistent with the inherited exponential
suppression expected from Corollary~\ref{cor:loss-classes}. The NLL row is
qualitatively different: the non-exponential candidates remain close to one
another, while the exponential model is separated by a larger $\Delta$AICc
penalty. The data do not identify a unique non-exponential law for NLL over
the accessible window, but they cleanly distinguish the NLL trend from the
exponential class governing linear and JSD.

\subsection{Synthesis}
\label{subsec:hydro_synthesis}

Taken together, the numerics support a precise but limited claim. Once the task is expressed through a polynomial-width compressed interface, the choice of classical head materially changes the gradient landscape: NLL produces substantially larger feature-space and transmitted gradients than linear and JSD, and this advantage survives finite-shot reliability and fidelity checks at the level of median circuit behavior.

At the same time, the present construction does not exhibit a qualitatively different shot complexity class. Over the explored range, the shot frontiers remain exponential-like for all three heads, with the non-affine head incurring a larger but broadly constant-factor overhead rather than a visibly different exponent. The mechanistic decomposition makes the obstruction explicit: model-side responsivity, not the choice of head, sets the shot frontier. An interface whose responsivity $\sigma_{\max}(J_{F,S})$ decays exponentially will fail the estimation constraint regardless of how favorable the loss-side geometry becomes, and this is the regime realized by the present probe.

The specific way responsivity collapses here points to a further nuance. Reducing the feature dimension does not by itself control the scale of the exposed observables: for fixed $b$, each block-weight statistic aggregates over $\Theta(n/b)$ qubits, so the exposed variables become increasingly extensive with system size even though $m(n)$ remains polynomial. Polynomial width is necessary but not sufficient; the scale and organization of the exposed variables likely matter as well. This in turn suggests transport-adapted mode-space interfaces, in which the exposed statistics are intensive linear functionals of the state at controlled wavelengths, as a natural direction for representations that might preserve responsivity while retaining task structure.

The numerical probe thus occupies a specific point in the design space of Section~\ref{sec:compressed_interface}: polynomial in width, rich enough for the amplification mechanism to become visible, but still estimation-limited because responsivity collapses on the chosen interface. Whether representations exist whose responsivity remains polynomial while preserving sufficient task structure is the open design question, stated as a hypothesis in Section~\ref{sec:discussion}.

%% file: figures/hydro_complexity_tables/multiparam_gradient_aicc_full.tex
\begin{tabular}{l c c c c c}
\hline
 & \textsc{poly} & \textsc{power-log} & \textsc{quasi-poly} & \textsc{exp} & AICc$_{\text{best}}$ \\

$\phi(n)$
& $\log n$
& $\log n + \log\log n$
& $(\log n)^2$
& $n$
& --- \\

scaling ($y \sim $)
&  $n^{\beta_1}$
&  $n^{\beta_1}(\log n)^{\beta_1}$
&  $e^{\beta_1(\log n)^2}$
&  $e^{\beta_1 n}$
& --- \\
\hline
linear
& 20.1
& 21.9
& 12.4
& \textbf{0}
& $-28.4$ \\

jsd
& 23.1
& 24.6
& 16.3
& \textbf{0}
& $-34.8$ \\

nll
& 0.06
& \textbf{0}
& 0.98
& 3.97
& $-25.7$ \\
\hline

\end{tabular}

%% file: 07-discussion-rev.tex
\section{Discussion}
\label{sec:discussion}

The obstruction was real, but local. Barren-plateau theorems prove exponential gradient suppression for losses that admit a fixed-observable representation relative to the measurement interface. Theorem \ref{thm:escape} shows this representation exists if and only if the loss is affine in the measured statistics. Beyond the affine regime, gradients are governed by the chain-rule decomposition of Section~\ref{sec:geometry}. The practical question becomes the representation-design problem of Section~\ref{sec:compressed_interface}: whether one can find polynomial-width interfaces on which responsivity, effective loss-side signal, and transmittance remain jointly favorable. The numerical demonstration of Section~\ref{sec:hydro_case_study} provides partial evidence for this mechanism: on a polynomial-width compressed interface, the amplification-capable head produces transmitted gradients whose scaling is statistically distinguishable from the exponential class governing both the affine and inheriting baselines.

We conjecture that the regime beyond this boundary is not empty.

\begin{hypothesis}[Polynomially-Barren \& Just-Right (PB\&J)]
\label{hyp:pbj}
We hypothesize that some physically and algorithmically natural learning tasks admit variational formulations of the form
\begin{equation*}
    L(\theta) = C\!\big(T(F(\rho(\theta)))\big),
\end{equation*}
where the feature map
$F:\mathcal{D}(\mathcal{H}_n)\to\mathbb{R}^{m(n)}$ has polynomial width
$m(n)=\mathrm{poly}(n)$, and there exists a set of random initializations
$A_n$ with
\begin{equation*}
    \Pr_{\theta}(A_n) \ge \frac{1}{\mathrm{poly}(n)}
\end{equation*}
such that for every $\theta \in A_n$:
\begin{enumerate}
    \item the effective loss-side signal is amplification-capable and remains at least inverse-polynomial on the relevant initialization set, i.e.
    \begin{equation*}
        \big\|g_F^{\mathrm{eff}}(\theta)\big\|
        \;\ge\;
        \frac{1}{\mathrm{poly}(n)}
        \qquad
        \text{for all }\theta\in A_n;
    \end{equation*}

    \item the feature-map Jacobian is not already flattened on the chosen
    interface,
    \begin{equation*}
        \sigma_{\max}\!\big(J_F(\theta)\big)
        \;\ge\;
        \frac{1}{\mathrm{poly}(n)};
    \end{equation*}

    \item the induced transmittance is not exponentially suppressed,
    \begin{equation*}
        \mathcal{T}(\theta)
        \;\ge\;
        \frac{1}{\mathrm{poly}(n)}.
    \end{equation*}
\end{enumerate}
Equivalently, on $A_n$ all three chain-rule factors are at worst
polynomially small. Under these conditions,
Theorem~\ref{thm:amplification} yields
\begin{equation*}
    \|\nabla_\theta L(\theta)\|
    \;\ge\;
    \frac{1}{\mathrm{poly}(n)}
    \qquad\text{for all }\theta \in A_n,
\end{equation*}
and hence
\begin{equation*}
    \mathbb{E}_\theta\!\left[\|\nabla_\theta L(\theta)\|\right]
    \;\ge\;
    \frac{1}{\mathrm{poly}(n)}.
\end{equation*}
Thus exponential barren-plateau suppression is not structurally enforced on
the chosen representation.
\end{hypothesis}

This regime lies outside the structural assumptions of existing barren-plateau proofs. The polynomial gradient bounds rule out a specific exponential structural obstruction, not all obstructions; they do not constitute guarantees of efficient optimization, convergence, or computational advantage. The binding remaining constraint in the numerical demonstration is model-side responsivity — a representation-level limitation. Effective theories offer a natural path forward: by construction, they identify low-dimensional representations that capture the relevant degrees of freedom at a given scale, precisely the combination of polynomial width and preserved responsivity that the PB\&J conditions require.
But the structural point stands: the barren-plateau theorems were precise, and their premises were explicit.

\bigskip
Now, the boundary is clear. The question returns to what it ought to be: which representations admit learning, and which do not.

%% file: Appendix/app-iff-rev.tex

\section{Affine objectives and fixed-observable structure}
\label{app:iff}

This appendix provides the formal classification underlying
Theorem~\ref{thm:escape}.  Fix a measurement \emph{interface} given by
linearly independent Hermitian observables $O_1,\dots,O_m$ and the
associated feature map
\[
  F(\rho) := \big(\tr(\rho\,O_1),\dots,\tr(\rho\,O_m)\big) \in \mathbb{R}^m.
\]
We consider objectives of the form $L(\theta) = f(F(\rho(\theta)))$, where
$\theta \mapsto \rho(\theta)$ is a differentiable ansatz and $f$ is differentiable
on an open set $\Omega \subset \mathbb{R}^m$ containing the relevant reachable
features.  Write $F(\theta) := F(\rho(\theta))$ and let
\[
  J_F(\theta) := \frac{\partial F}{\partial \theta}(\theta)
  \in \mathbb{R}^{m \times P}
\]
denote the feature Jacobian.  We say that the ansatz \emph{explores} an open region
$U \subset \Omega$ if $U \subset \{F(\theta) : \theta \in \mathbb{R}^P\}$.

\subsection{Fixed-observable structure relative to an interface}

\begin{definition}[Fixed-observable structure relative to an interface]
\label{def:fixed_struct}
We say that $L(\theta) = f(F(\rho(\theta)))$ admits a \emph{fixed-observable
representation on $U$ relative to the interface $\{O_j\}$} if there exist
$H \in \mathrm{span}\{O_1,\dots,O_m\}$ and $c \in \mathbb{R}$ such that
\[
  L(\theta) \;=\; \tr\!\big(H\,\rho(\theta)\big) + c
\]
for all $\theta$ with $F(\theta) \in U$.
\end{definition}

The restriction $H \in \mathrm{span}\{O_j\}$ is the operational content of ``relative
to the interface'': it captures the idea that the objective depends on the quantum state
only through the a priori measured statistics $F(\rho)$.

\subsection{Proof of Theorem~\ref{thm:escape}}

Theorem~\ref{thm:escape} states that, on any non-empty open set $U$ explored by
the ansatz, a fixed-observable representation exists if and only if $f$ is affine on
$U$.  The main text proves the non-affine $\Rightarrow$ no-fixed-observable direction
by contradiction.  Here we record both directions for completeness.

\begin{proof}[Proof of Theorem~\ref{thm:escape}]

\emph{Affine $\Rightarrow$ fixed observable.}\;
If $f(F) = a^\top F + c$ on $U$, define
$H := \sum_{j=1}^m a_j O_j \in \mathrm{span}\{O_j\}$.
Then for any $\theta$ with $F(\theta) \in U$,
\[
  L(\theta) = f(F(\theta))
  = c + \sum_{j=1}^m a_j\,\tr(\rho(\theta)\,O_j)
  = c + \tr\!\big(H\,\rho(\theta)\big).
\]

\emph{Fixed observable $\Rightarrow$ affine.}\;
Suppose $H \in \mathrm{span}\{O_j\}$ and $c \in \mathbb{R}$ satisfy
$L(\theta) = \tr(H\,\rho(\theta)) + c$ for all $\theta$ with $F(\theta) \in U$.
By linear independence of $\{O_j\}$, $H = \sum_{j=1}^m a_j O_j$ with unique
coefficients $a \in \mathbb{R}^m$.  Then for every such $\theta$,
\[
  f(F(\theta))
  = L(\theta)
  = a^\top F(\theta) + c.
\]
Since the ansatz explores $U$, every $F \in U$ is attained, so $f(F) = a^\top F + c$
for all $F \in U$.  That is, $f$ is affine on $U$.
\end{proof}

\begin{remark}[ Sufficient condition for the exploration hypothesis]
\label{rem:full_row_rank}
If $J_F(\theta)$ has full row rank $m$ on an open parameter region, then
$F$ is a submersion there, and the ansatz explores an open set of feature space.
\end{remark}

\subsection{Scope: fixed versus parameter-dependent observables}

\begin{remark}[ Parameter-dependent observables are post hoc]
\label{rem:theta_dependent_obs}
Concentration-based barren-plateau arguments are formulated for \emph{fixed}
(parameter-independent) operators: one studies expectations of fixed observables
with respect to randomly initialized circuit states.  Allowing observables that
depend explicitly on $\theta$ defines a different regime.  Indeed, for any
differentiable objective one can always construct a $\theta$-dependent operator
$H_{\mathrm{eff}}(\theta)$ that reproduces a given local gradient via
$\partial_{\theta_k}L(\theta) =
\tr(H_{\mathrm{eff}}(\theta)\,\partial_{\theta_k}\rho(\theta))$.
Such constructions are necessarily post hoc: the operator co-evolves with the
objective and encodes the full difficulty of the optimization problem.  They
therefore do not constitute counterexamples to concentration-based barren-plateau
theory, whose proofs rely essentially on $\theta$-independent operators.
\end{remark}

\subsection{Connection to concentration-based barren-plateau proofs}

\begin{remark}[ When barren-plateau machinery applies]
\label{rem:bp_connection}
Theorem~\ref{thm:escape} identifies exactly when the objective itself has
fixed-observable structure relative to a chosen measurement interface:
precisely in the affine regime.  In that regime, standard gate-based identities
(such as parameter-shift rules, when available) express gradient components
through fixed operators, and the usual concentration-based barren-plateau
proof template applies.  Outside that regime, the objective no longer fits
that fixed-observable template, so gradient behavior must be analyzed through
the chain-rule geometry developed in Section~\ref{sec:geometry}.
\end{remark}

%% file: Appendix/app-elliptical.tex

\section{Elliptical refinements for transmittance under compressed affine interfaces}
\label{app:elliptical_refinements}

In the main text, our ``random orientation'' baseline for transmittance is intentionally minimal
and should be read as an initialization-time yardstick: when one of the two directions in feature
space is (Euclidean) isotropic relative to the other, the typical overlap is of order $1/\sqrt{m}$.
After introducing a structured compressed interface $F(\rho)\in\mathbb{R}^m$ and a classical head,
however, Euclidean isotropy in the feature coordinates is not generic and should not be expected to
persist. This appendix records a simple anisotropic refinement that remains analyzable without
committing to task-specific structure. The refinement applies when the compression/exposure is
\emph{affine in operator statistics} (hence governed by linear maps on a high-dimensional isotropic
source), in which case isotropic geometry in the source induces \emph{elliptical} geometry in feature
space. For genuinely non-affine interfaces, comparable quantitative baselines would require additional problem-specific priors and are outside our scope.

\subsection{An elliptically contoured null model and the effective dimension}
\label{app:elliptical_null_model}

We model anisotropy in feature space by a positive semidefinite ``shape'' matrix
$\Sigma\succeq 0$ on $\mathbb{R}^m$. A canonical way to generate an elliptical direction is to take an
isotropic vector and apply a linear map: if $z\sim\mathcal{N}(0,I_m)$ and $x=\Sigma^{1/2}z$, then
$x$ has covariance $\Sigma$. Normalizing yields a random unit direction supported on an ellipsoid:
\begin{equation}
    u \;=\; \frac{\Sigma^{1/2} z}{\sqrt{z^\top \Sigma z}}
    \;\in\; S^{m-1}.
    \label{eq:elliptical_direction}
\end{equation}
When $\Sigma\propto I_m$, this reduces to a uniform direction on $S^{m-1}$, recovering the isotropic
baseline. Under anisotropy, different directions receive different weights according to the spectrum
of $\Sigma$.

In this appendix we focus on the geometric statistic underlying transmittance, namely the overlap of
two directions. Let $u,v\in S^{m-1}$ be independent elliptical directions generated as in
\eqref{eq:elliptical_direction} from independent $z,z'\sim\mathcal{N}(0,I_m)$. Define the RMS overlap
scale
\begin{equation}
    \mathcal{T}_{\mathrm{rms}}(\Sigma)
    \;:=\;
    \Big(\mathbb{E}\big[(u^\top v)^2\big]\Big)^{1/2}.
    \label{eq:Trms_def}
\end{equation}
(We use an RMS notion because it admits a clean expression in terms of second moments; it can be converted to statements about $\mathbb{E}|u^\top v|$ under additional distributional regularity.)

The quantity governing \eqref{eq:Trms_def} is the scale-invariant ``effective dimension''
\begin{equation}
    d_{\mathrm{eff}}(\Sigma)
    \;:=\;
    \frac{\mathrm{Tr}(\Sigma)^2}{\mathrm{Tr}(\Sigma^2)}
    \;=\;
    \frac{\|\lambda\|_1^2}{\|\lambda\|_2^2},
    \label{eq:deff_def}
\end{equation}
where $\lambda=(\lambda_1,\dots,\lambda_m)$ are the eigenvalues of $\Sigma$. Note that
$1\le d_{\mathrm{eff}}(\Sigma)\le \mathrm{rank}(\Sigma)\le m$, with $d_{\mathrm{eff}}(\Sigma)=m$ in
the isotropic case $\Sigma\propto I_m$, and $d_{\mathrm{eff}}(\Sigma)=r$ for a rank-$r$ spectrum with
equal nonzero eigenvalues.

\begin{lemma}[RMS overlap under an elliptical null model (proof sketch)]
\label{lem:elliptical_overlap_rms}
Let $u,v$ be generated by \eqref{eq:elliptical_direction} from independent $z,z'\sim\mathcal{N}(0,I_m)$.
Then
\begin{equation}
    \mathbb{E}\big[(u^\top v)^2\big]
    \;\approx\;
    \frac{\mathrm{Tr}(\Sigma^2)}{\mathrm{Tr}(\Sigma)^2}
    \;=\;
    \frac{1}{d_{\mathrm{eff}}(\Sigma)},
    \label{eq:elliptical_overlap_rms}
\end{equation}
whenever the quadratic forms $z^\top \Sigma z$ and $z'^\top \Sigma z'$ concentrate sharply around
$\mathrm{Tr}(\Sigma)$. In particular, in regimes where $\mathrm{Tr}(\Sigma^2)\ll \mathrm{Tr}(\Sigma)^2$
(i.e.\ $d_{\mathrm{eff}}(\Sigma)\gg 1$), the typical overlap scale satisfies
$\mathcal{T}_{\mathrm{rms}}(\Sigma)\asymp 1/\sqrt{d_{\mathrm{eff}}(\Sigma)}$ up to constant factors.
\end{lemma}

\begin{proof}[Proof sketch]
We denote
\[
u^\top v
=
\frac{z^\top \Sigma z'}{\sqrt{(z^\top \Sigma z)(z'^\top \Sigma z')}}.
\]
The numerator has mean zero and
\[
z^\top \Sigma z'=\sum_{i,j}\Sigma_{ij} z_i z'_j.
\]
Using independence of $z$ and $z'$ together with
$\mathbb E[z_i z_k]=\delta_{ik}$ and $\mathbb E[z'_j z'_l]=\delta_{jl}$,
we obtain
\[
\mathbb E\big[(z^\top \Sigma z')^2\big]
=
\sum_{i,j,k,l}\Sigma_{ij}\Sigma_{kl}\,\delta_{ik}\delta_{jl}
=
\sum_{i,j}\Sigma_{ij}^2
=
\mathrm{Tr}(\Sigma^\top\Sigma)
=
\mathrm{Tr}(\Sigma^2),
\]
where the last equality uses symmetry of $\Sigma$. The denominator terms are weighted chi-square quadratic forms with mean
$\mathbb{E}[z^\top\Sigma z]=\mathrm{Tr}(\Sigma)$ and variance
$\mathrm{Var}(z^\top\Sigma z)=2\,\mathrm{Tr}(\Sigma^2)$. Standard concentration for such quadratic
forms implies $z^\top\Sigma z=\mathrm{Tr}(\Sigma)\,(1+o(1))$ (and similarly for $z'$) whenever
$\mathrm{Tr}(\Sigma^2)/\mathrm{Tr}(\Sigma)^2$ is small. Substituting these typical values yields
\eqref{eq:elliptical_overlap_rms}.
\end{proof}

The lemma motivates $d_{\mathrm{eff}}(\Sigma)$ as the correct replacement for the ambient dimension
$m$ when Euclidean isotropy fails but second-moment geometry remains meaningful.

\subsection{A spectral ``menu'' and its implications for transmittance}
\label{app:elliptical_menu}

Table~\ref{tab:elliptical_menu} previews several common spectral shapes for $\Sigma$ and the
corresponding behavior of $d_{\mathrm{eff}}(\Sigma)$ and $\mathcal{T}_{\mathrm{rms}}(\Sigma)$. The
examples should be read as modeling templates: they do not assert that a given interface produces a
particular spectrum, but rather show how different coarse-graining designs can interpolate between a
near-isotropic regime ($d_{\mathrm{eff}}\approx m$) and a low-effective-dimension regime
($d_{\mathrm{eff}}\ll m$).

\begin{table}[t]
\centering
\small
\renewcommand{\arraystretch}{1.4}
\begin{tabular}{llll}
\hline\hline
\textbf{Spectral shape} &
\textbf{Eigenvalues (up to scale)} &
\textbf{Eff.\ dimension} &
\textbf{RMS overlap} \\
\hline

Well-conditioned &
$\lambda_i \in [\lambda_{\min},\lambda_{\max}]$ &
$m/\kappa \lesssim d_{\mathrm{eff}} \le m$ &
$\lesssim \sqrt{\kappa/m}$ \\

Low-rank &
$\lambda_1=\cdots=\lambda_r>0$, rest $0$ &
$d_{\mathrm{eff}} = r$ &
$\asymp 1/\sqrt{r}$ \\

Spiked + bulk &
$\lambda_1 = 1+s$, rest $=1$ &
$(m+s)^2/(m+2s+s^2)$ &
$1/\sqrt{m}$ to $O(1)$ \\

Power-law &
$\lambda_i \propto i^{-\alpha}$ &
$m$ to $\Theta(1)$ &
$\asymp 1/\sqrt{d_{\mathrm{eff}}}$ \\

Exponential &
$\lambda_i \propto e^{-ci}$ &
$\Theta(1)$ &
$\Theta(1)$ \\

Block-structured &
block-diagonal $\Sigma$ &
$\sum_b d_{\mathrm{eff}}^{(b)}$ &
block-dependent \\

\hline\hline
\end{tabular}
\caption{Spectral menu for anisotropic null models.}
\label{tab:elliptical_menu}
\end{table}
\subsection{Derivations behind Table~\ref{tab:elliptical_menu}}
\label{app:elliptical_derivations}

\paragraph{Well-conditioned spectra.}
Assume $\lambda_i\in[\lambda_{\min},\lambda_{\max}]$ and define $\kappa=\lambda_{\max}/\lambda_{\min}$.
Using $\mathrm{Tr}(\Sigma^2)=\sum_i \lambda_i^2 \le \lambda_{\max}\sum_i \lambda_i = \lambda_{\max}\mathrm{Tr}(\Sigma)$,
we obtain
\[
d_{\mathrm{eff}}(\Sigma)
=
\frac{\mathrm{Tr}(\Sigma)^2}{\mathrm{Tr}(\Sigma^2)}
\ge
\frac{\mathrm{Tr}(\Sigma)}{\lambda_{\max}}
\ge
\frac{m\lambda_{\min}}{\lambda_{\max}}
=
\frac{m}{\kappa}.
\]
The upper bound $d_{\mathrm{eff}}(\Sigma)\le m$ holds generally. By Lemma~\ref{lem:elliptical_overlap_rms},
$\mathcal{T}_{\mathrm{rms}}\asymp 1/\sqrt{d_{\mathrm{eff}}}$, yielding $\mathcal{T}_{\mathrm{rms}}\lesssim \sqrt{\kappa/m}$.

\paragraph{Low-rank spectra.}
If $\lambda_1=\cdots=\lambda_r>0$ and the rest vanish, then
$\mathrm{Tr}(\Sigma)=r\lambda_1$ and $\mathrm{Tr}(\Sigma^2)=r\lambda_1^2$, hence
$d_{\mathrm{eff}}(\Sigma)=r$ and $\mathcal{T}_{\mathrm{rms}}\asymp 1/\sqrt{r}$.

\paragraph{Spiked + bulk.}
For $\lambda_1=1+s$ and $\lambda_2=\cdots=\lambda_m=1$,
\[
\mathrm{Tr}(\Sigma)=m+s,
\qquad
\mathrm{Tr}(\Sigma^2)=(1+s)^2+(m-1)\cdot 1 = m+2s+s^2,
\]
so $d_{\mathrm{eff}}(\Sigma)=(m+s)^2/(m+2s+s^2)$ and
$\mathcal{T}_{\mathrm{rms}}\asymp \sqrt{(m+2s+s^2)/(m+s)^2}$.
This explicitly interpolates between the isotropic regime ($s\ll m$) and the strongly spiked regime
($s\gg m$).

\paragraph{Power-law decay.}
Let $\lambda_i \propto i^{-\alpha}$ and define $S_1(m)=\sum_{i=1}^m i^{-\alpha}$ and
$S_2(m)=\sum_{i=1}^m i^{-2\alpha}$. Then $d_{\mathrm{eff}}(\Sigma)\asymp S_1(m)^2/S_2(m)$.
Standard integral comparisons give the asymptotics:
\[
S_1(m)\asymp
\begin{cases}
m^{1-\alpha}, & \alpha<1,\\
\log m, & \alpha=1,\\
1, & \alpha>1,
\end{cases}
\qquad
S_2(m)\asymp
\begin{cases}
m^{1-2\alpha}, & \alpha<\tfrac12,\\
\log m, & \alpha=\tfrac12,\\
1, & \alpha>\tfrac12.
\end{cases}
\]
Combining yields the piecewise scaling in Table~\ref{tab:elliptical_menu}. In particular, for
$\tfrac12<\alpha<1$ we obtain $d_{\mathrm{eff}}(\Sigma)\asymp m^{2-2\alpha}$, while for $\alpha>1$
both sums converge and $d_{\mathrm{eff}}(\Sigma)=\Theta(1)$.

\paragraph{Exponential decay.}
If $\lambda_i \propto e^{-c i}$, then both $\sum_i \lambda_i$ and $\sum_i \lambda_i^2$ converge
geometrically as $m\to\infty$, hence $d_{\mathrm{eff}}(\Sigma)=\Theta(1)$ and
$\mathcal{T}_{\mathrm{rms}}=\Theta(1)$.

\paragraph{Block-structured (locality-induced) spectra.}
In many physically motivated interfaces, features decompose into multiple weakly coupled blocks,
such as local patches, correlator groups, or low-order marginals supported on bounded regions.
If the induced shape matrix $\Sigma$ is approximately block-diagonal,
$\Sigma \approx \mathrm{diag}(\Sigma^{(1)},\Sigma^{(2)},\ldots)$, then the effective dimension
decomposes additively,
\[
d_{\mathrm{eff}}(\Sigma)\;\approx\;\sum_b d_{\mathrm{eff}}(\Sigma^{(b)}),
\]
up to weighting by the relative spectral mass of each block.
This regime naturally interpolates between low-rank and near-isotropic behavior and is
characteristic of thermal states, locally constrained objectives, and correlator-based
interfaces where information is distributed across many small subsystems rather than
concentrated in a single global mode.

From a design perspective, the practitioner’s task is to identify which spectral regime a given
interface–head pair $(F,f)$ induces—near-isotropic, low-rank, spiked, power-law, or block-structured—
and then read off the corresponding $d_{\mathrm{eff}}$ scaling as a guide to the expected
transmittance behavior.

\subsection{Scope, limitations, and a two-shape remark}
\label{app:elliptical_scope}

This appendix is intended as a refinement for settings where Euclidean isotropy of the loss-side
direction is not a credible baseline after choosing a structured compressed interface and head.
When Euclidean isotropy does hold (as in the one-sided isotropy null used in the main text), the
$1/\sqrt{m}$ overlap scale is recovered regardless of model-side anisotropy, and elliptical
refinements are not needed. Conversely, when the interface/head induces anisotropy, elliptical null
models provide a maximally general, second-moment-controlled baseline for overlap; beyond affine
(exposure) structure, such second-moment descriptions typically fail and quantitative predictions
require task-specific priors or local approximations.

\begin{remark}[ Two anisotropic shapes]
\label{rem:two_shape_elliptical}
The preceding discussion assumed a common shape matrix $\Sigma$ for both directions. A slightly more
general model assigns distinct shapes $\Sigma_u,\Sigma_v\succeq 0$ and samples
$u=\Sigma_u^{1/2}z/\sqrt{z^\top\Sigma_u z}$ and $v=\Sigma_v^{1/2}z'/\sqrt{z'^\top\Sigma_v z'}$.
Under the same quadratic-form concentration heuristic as in Lemma~\ref{lem:elliptical_overlap_rms},
one obtains the RMS approximation
\[
\mathbb{E}\big[(u^\top v)^2\big]
\;\approx\;
\frac{\mathrm{Tr}(\Sigma_u\Sigma_v)}{\mathrm{Tr}(\Sigma_u)\,\mathrm{Tr}(\Sigma_v)}.
\]
This suggests an ``overlap parameter'' between anisotropies, and recovers
$\mathrm{Tr}(\Sigma^2)/\mathrm{Tr}(\Sigma)^2$ when $\Sigma_u=\Sigma_v=\Sigma$.
We do not develop this refinement further in the main text, but it can be useful when separately
modeling head-induced anisotropy and interface-induced anisotropy.
\end{remark}

%% file: Appendix/app-04-experiment-details.tex
\section{Numerical simulations and supplementary diagnostics}
\label{app:frontier_defs}

This appendix collects numerical details supporting the case
study of Section~\ref{sec:hydro_case_study}. We first provide the formal
definitions of the finite-shot reliability and fidelity statistics used to
define the accepted shot frontiers. We then give a supplementary diagnostic
for the broad finite-shot spread of the NLL estimator in the full-subspace
theorem-aligned run.

\subsection{Setup and notation}

Fix a system size~$n$, a classical head (linear, JSD, or NLL), and a shot
budget~$M$. Across all reported system sizes, the teacher target distribution
is sampled at a fixed budget of $2 \times 10^5$ shots, so that teacher
estimation does not introduce an additional size-dependent variable into the
comparison.

For each circuit~$c$, let $r = 1, \ldots, R$ index independent finite-shot
repetitions at budget~$M$, each producing a gradient estimate. We write
$\bar{g}_c(M)$ for the repetition-averaged gradient of circuit~$c$:
\begin{equation*}
    \bar{g}_c(M)
    \;:=\;
    \frac{1}{R} \sum_{r=1}^{R} \hat{g}_{c,r}(M).
\end{equation*}
In the single-parameter setting, $\hat{g}_{c,r}(M) \in \mathbb{R}$ is a scalar
directional derivative. In the multi-parameter setting,
$\hat{g}_{c,r}(M) \in \mathbb{R}^s$ is a subspace gradient vector over $s$
coordinates.

\subsection{Single-parameter signal-to-noise ratio}

The canonical single-parameter probe uses $C = 200$ independent random circuit
instances. Each circuit is evaluated at $R = 30$ independent finite-shot
repetitions per shot budget; all reported frontier points $M^*(n)$ fall within
this regime.

For each circuit~$c$, define the circuit-level signal-to-noise ratio
\begin{equation*}
    \mathrm{SNR}^{\mathrm{single}}_c(M)
    \;:=\;
    \frac{|\bar{g}_c(M)|}
         {\sqrt{\widehat{\mathrm{Var}}_r\!\bigl(\hat{g}_{c,r}(M)\bigr)}},
\end{equation*}
where $\widehat{\mathrm{Var}}_r$ denotes the sample variance across
repetitions. The numerator measures the magnitude of the averaged signal; the
denominator measures the scale of shot-to-shot fluctuation for that circuit.
The ensemble summary is the median across circuits:
\begin{equation*}
    \mathrm{MedSNR}^{\mathrm{single}}(M)
    \;:=\;
    \operatorname{median}_{c \in [C]}\;
    \mathrm{SNR}^{\mathrm{single}}_c(M).
\end{equation*}
We use the median rather than the mean to ensure that the summary reflects
typical-circuit behavior and is not dominated by a small number of
high-variance or high-signal outliers.

The single-parameter accepted frontier is
\begin{equation}
\label{eq:frontier_single}
    M^{*}_{\mathrm{single}}(n)
    \;:=\;
    \min\bigl\{
        M :\;
        \mathrm{MedSNR}^{\mathrm{single}}(M) \ge \kappa
    \bigr\},
\end{equation}
with $\kappa = 2$ in all reported experiments. This threshold ensures that, for a
typical circuit, the averaged directional gradient is at least twice its
shot-noise scale.

\subsection{Multi-parameter signal-to-noise ratio}

The full-subspace study uses a $C = 60$ circuit ensemble with a fixed subspace
size $s = 32$, and evaluates the finite-shot frontier using $R = 200$
independent repetitions per circuit/subspace group.

In the subspace setting, the gradient estimate is vector-valued:
$\hat{g}_{c,r}(M) \in \mathbb{R}^s$. The circuit-level signal-to-noise ratio
generalizes to
\begin{equation*}
    \mathrm{SNR}^{\mathrm{multi}}_c(M)
    \;:=\;
    \frac{\|\bar{g}_c(M)\|_2}
         {\sqrt{\displaystyle\sum_{j=1}^{s}
           \widehat{\mathrm{Var}}_r\!\bigl(\hat{g}_{c,r,j}(M)\bigr)}},
\end{equation*}
where the denominator aggregates the coordinate-wise shot-noise variances into
a scalar noise scale for the full subspace object. The ensemble summary is
again the median:
\begin{equation*}
    \mathrm{MedSNR}^{\mathrm{multi}}(M)
    \;:=\;
    \operatorname{median}_{c \in [C]}\;
    \mathrm{SNR}^{\mathrm{multi}}_c(M).
\end{equation*}

\subsection{Multi-parameter relative bias}

Signal-to-noise measures whether the gradient estimate is \emph{resolved}
relative to its own shot noise. It does not measure whether the estimate is
\emph{faithful} to the true gradient---a resolved estimate can be
systematically biased if, for example, the non-linear head distorts the
gradient under finite sampling. To guard against this, we introduce a
complementary fidelity statistic.

Let $g_c^{\mathrm{exact}} \in \mathbb{R}^s$ denote the exact subspace gradient
for circuit~$c$, obtained from the saved exact chain-rule sketch payload
(exact conditional on the fixed sampled teacher target used throughout the
study). The circuit-level relative bias is
\begin{equation*}
    \mathrm{RelBias}_c(M)
    \;:=\;
    \frac{\|\bar{g}_c(M) - g_c^{\mathrm{exact}}\|_2}
         {\|g_c^{\mathrm{exact}}\|_2 + \varepsilon},
\end{equation*}
where $\varepsilon > 0$ is a small stabilizing constant that prevents division
by zero for circuits whose exact gradient is near-vanishing. The ensemble
summary is
\begin{equation*}
    \mathrm{MedRelBias}(M)
    \;:=\;
    \operatorname{median}_{c \in [C]}\;
    \mathrm{RelBias}_c(M).
\end{equation*}
A value of $\mathrm{MedRelBias}(M) \le \tau$ ensures that, for a typical
circuit, the finite-shot subspace gradient deviates from the exact value by at
most a fraction~$\tau$ of its magnitude.

\subsection{Multi-parameter accepted frontier}

The multi-parameter frontier imposes both reliability and fidelity
simultaneously:
\begin{equation}
\label{eq:frontier_multi}
    M^{*}_{\mathrm{multi}}(n)
    \;:=\;
    \min\bigl\{
        M :\;
        \mathrm{MedSNR}^{\mathrm{multi}}(M) \ge \kappa,
        \;\;
        \mathrm{MedRelBias}(M) \le \tau
    \bigr\},
\end{equation}
with $\kappa = 2$ and $\tau = 0.5$ in all reported experiments. The SNR
condition ensures that the subspace gradient is resolved above shot noise; the
RelBias condition ensures that the resolved estimate has not drifted
substantially from its exact counterpart. The joint condition is strictly
more demanding than either criterion alone.

\subsection{Choice of thresholds}

The thresholds $\kappa = 2$ and $\tau = 0.5$ are operational conventions, not fundamental constants. $\kappa = 2$ requires the signal to be twice the noise
scale, ensuring the directional gradient can be clearly resolved above shot-to-shot fluctuation. $\tau = 0.5$ permits up to 50\%
relative deviation from the exact gradient, which is deliberately permissive:
the goal is to exclude gross distortion while remaining compatible with the
noise levels inherent in moderate shot budgets. The qualitative conclusions of Section~\ref{sec:hydro_case_study} — specifically, the ordering of gradient magnitudes across heads and the scaling classification — appear stable under moderate threshold variation in our tests.

\subsection{Distributional structure of the NLL finite-shot estimator}
\label{app:ridgeline}

The $q25/q75$ bars in Fig.~\ref{fig:hydro_chainrule_theorem} compress a
genuinely broad and non-Gaussian cross-circuit distribution for the NLL
finite-shot estimator. To make this explicit, Fig.~\ref{fig:appendix_ridgeline_nll}
shows a ridgeline view of the per-circuit finite-shot distributions at the
same accepted joint frontier $M^*(n)$ used in the main-text theorem figure.

The key point is interpretive. The broad NLL uncertainty in the theorem-aligned decomposition is not well described as a narrow unimodal cloud around the exact median. Instead, the finite-shot distribution develops visible multimodality on the 60-circuit ensemble, especially on the transmitted-gradient quantity. This is why the median remains informative for typical-circuit behavior while the interquartile spread becomes much broader than for the linear and JSD baselines.

\begin{figure}[t]
    \centering
    \includegraphics[width=\linewidth]{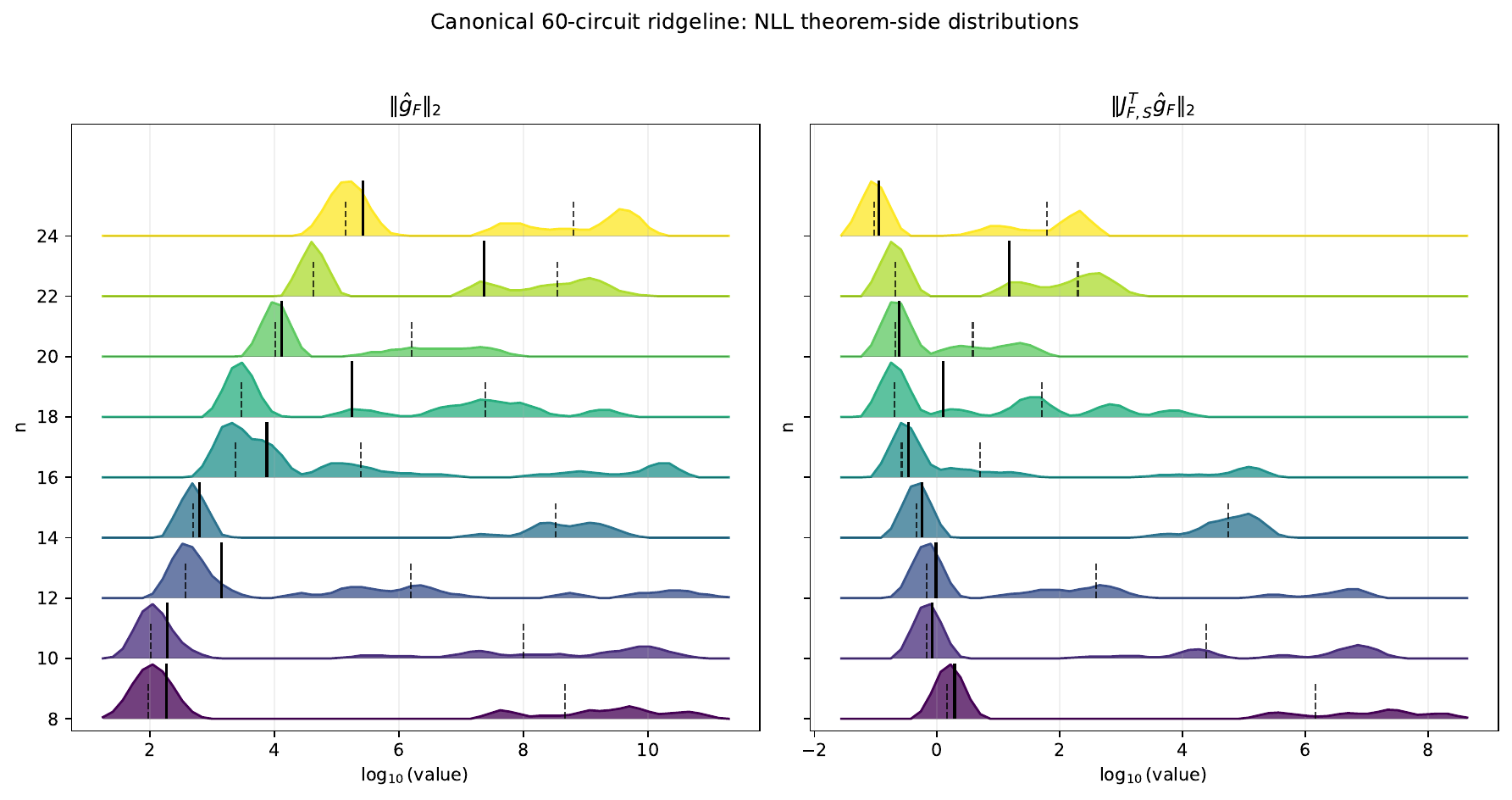}
    \caption{Ridgeline view of the per-circuit finite-shot NLL
    theorem-estimator distribution on the canonical 60-circuit full-subspace
    run for $n=8,\dots,24$ at the same joint frontier $M^*(n)$
    ($\mathrm{MedSNR} \ge 2$, $\mathrm{MedRelBias} \le 0.5$, $b=4$).
    Densities are shown on a $\log_{10}$ scale for $\|\hat g_F\|$ and
    $\|J_{F,S}^T \hat g_F\|$. The multimodal structure explains why NLL can
    show modest lower-quartile behavior while the median and upper quartile
    broaden sharply at some system sizes.}
    \label{fig:appendix_ridgeline_nll}
\end{figure}

\subsection{Multi-parameter shot frontiers}
\label{app:shot_frontiers_multi}
Fig.~\ref{fig:appendix_shotfrontier_multi} reports the accepted multi-parameter shot frontiers $M^*(n)$ used in the theorem-aligned decomposition of Section~\ref{sec:hydro_case_study}. The frontiers are obtained under the joint criterion of Eq.~\eqref{eq:frontier_multi}, which requires both $\mathrm{MedSNR}^{\mathrm{multi}} \geq 2$
and $\mathrm{MedRelBias} \leq 0.5$ across the $C = 60$ circuit ensemble at each system size. The absolute frontier values and the NLL/linear ratios are reported in panels (A) and (B) respectively.

\begin{figure*}[ht]
    \centering
    \includegraphics[width=\textwidth]{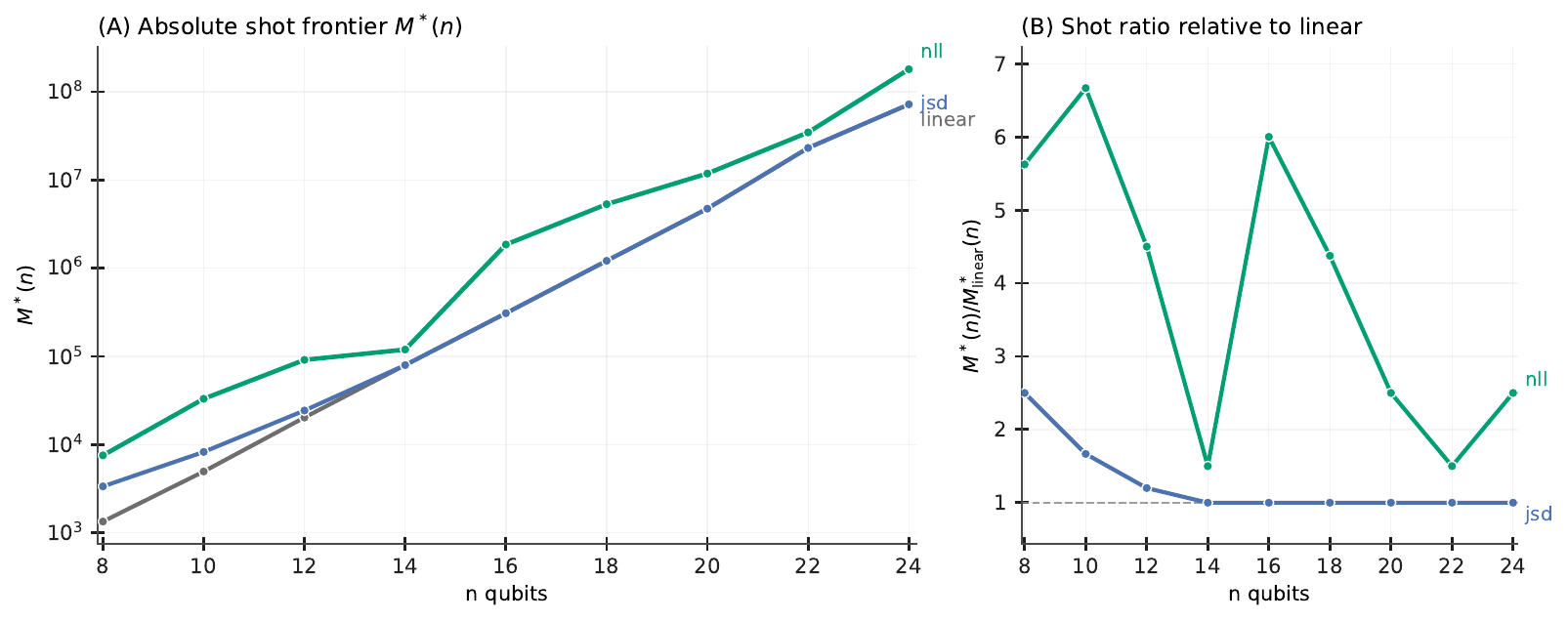}
    \caption{Accepted shot frontiers $M^*_{\mathrm{multi}}(n)$ for the $s=32$ multi-parameter probe at $b=4$, using the joint thresholds $\kappa=2$ and $\tau=0.5$. Panel~(A): absolute shot frontier. Panel~(B): shot ratio relative to the linear baseline. Linear and JSD exhibit essentially the same shot budgets from $n\ge 14$ onward. NLL incurs a variable overhead across the explored range ($1.5$--$6.7\times$), larger and less stable than in the single-parameter probe of Fig.~\ref{fig:hydro_singleparam_v21}. The non-monotone structure is consistent with the multimodal NLL estimator distribution documented in Appendix~\ref{app:ridgeline}. Despite this overhead, the NLL/linear ratio does not grow systematically with $n$, indicating that all three heads remain in the same broad exponential shot-complexity class.}
    \label{fig:appendix_shotfrontier_multi}
\end{figure*}

The multi-parameter frontier exhibits two distinct regimes. For linear and JSD, the shot budgets are essentially indistinguishable from $n=14$ onward, consistent with both heads occupying the same scaling class with matched constants. For NLL, the overhead is larger than in the single-parameter probe (Fig.~\ref{fig:hydro_singleparam_v21}) and exhibits non-monotone variation between $1.5\times$ and $6.7\times$. This variation is driven by the multimodal cross-circuit distribution of the NLL finite-shot estimator, documented in Appendix~\ref{app:ridgeline}: when the finite-shot distribution has heavy or bimodal upper tails, the median reliability criterion can be satisfied at dramatically different shot budgets at nearby system sizes. The important structural observation is that the ratio does not grow systematically with $n$, so all three heads remain in the same exponential shot-complexity class even though NLL's constants are larger and noisier.

\newpage

\subsection{Sensitivity to block count}
\label{app:b6_robustness}

To check the sensitivity of the charge-conserving joint-block results to the
choice of block count, we repeated the single-parameter and full-subspace
probes at $b=6$. The protocol is otherwise unchanged: the same teacher and
student families, the same depth schedules, the same head definitions, and the
same single- and multi-parameter frontier criteria are used throughout.

\begin{figure*}[ht!]
    \centering
    \includegraphics[width=\textwidth]{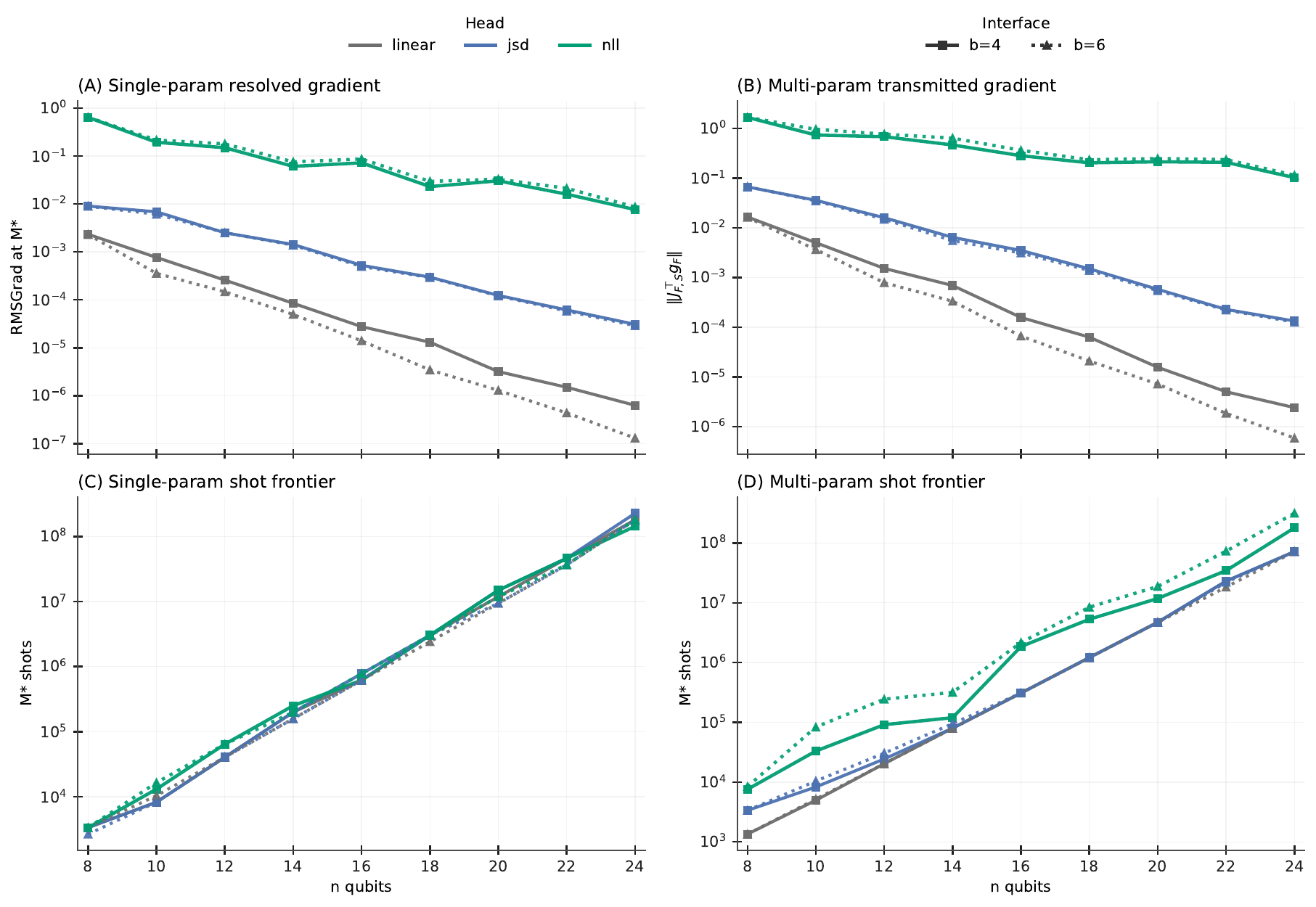}
    \caption{Comparison of the charge-conserving joint-block results for block counts $b=4$ and $b=6$. Solid square markers denote $b=4$ and dotted triangle markers denote $b=6$; colors indicate the classical head. Panel~(A): single-parameter resolved gradient at the accepted frontier $M^*$. Panel~(B): exact transmitted subspace gradient norm $\|J_{F,S}^{\top} g_F\|$. Panel~(C): single-parameter shot frontier $M^*$. Panel~(D): multi-parameter shot frontier $M^*$. Across both probes, the qualitative head ordering and overall scaling trends remain similar for $b=4$ and $b=6$, indicating that the conclusions of Section~\ref{sec:hydro_case_study} are not sensitive to this change in block count over the explored range.}
    \label{fig:appendix_b4_b6_robustness}
\end{figure*}

The comparison in Fig.~\ref{fig:appendix_b4_b6_robustness} leaves the main
interpretation unchanged. NLL continues to produce the largest resolved and
transmitted gradients, while the shot frontiers remain in the same broad
scaling regime as the affine and inheriting baselines. Over the explored
system-size window, changing the block count from $b=4$ to $b=6$ therefore
produces only modest constant-level shifts rather than a qualitative change in
behavior.